\documentclass[journal]{IEEEtran}
\usepackage{epsfig,amsmath,booktabs,dcolumn}
\usepackage{multirow}
\usepackage{tikz}
\usepackage{url}
\usepackage{cite}

\usepackage[draft]{hyperref}
\usepackage{lipsum}
\usepackage{setspace}
\usepackage{booktabs, tabularx}
\usepackage{array, xspace, subfigure, url, xcolor, amsmath, bm, array,amsfonts, balance,multirow, amsthm}

\usepackage[linesnumbered,ruled,vlined]{algorithm2e}
\usepackage[noend]{algpseudocode}

\SetCommentSty{mycommfont}

\newcommand{\ie}{\emph{i.e.,}\xspace}
\newcommand{\eg}{\emph{e.g.,}\xspace}
\newcommand{\sys}{Umbrella\xspace}

\newcommand{\first}{\textsf{(i)}\xspace}
\newcommand{\second}{\textsf{(ii)}\xspace}
\newcommand{\third}{\textsf{(iii)}\xspace}

\usepackage{mathtools}

\newtheorem{theorem}{Theorem}
\newtheorem{lemma}{Lemma}

\usepackage[font={small,bf}]{caption}

\begin{document}

\title{\huge \sys: Enabling ISPs to Offer Readily Deployable and Privacy-Preserving DDoS Prevention Services}

\author{Zhuotao~Liu,
        Yuan~Cao,~Min~Zhu,
        and~Wei~Ge
\thanks{Z. Liu is with Electrical and Computer Engineering Department, University of Illinois at Urbana-Champaign, Urbana, IL 61801 USA. Z. Liu is also with Google Inc., Mountain View, CA 94043 USA (email: zliu48@illinois.edu).}
\thanks{Y. Cao is with the College of Internet of Things Engineering, Hohai University, Changzhou 213022, China. Y. Cao is also with the Key Laboratory of Computer Network and Information Integration, Ministry of Education, Southeast University (Corresponding author: Yuan Cao, E-mail: 20161965@hhu.edu.cn).}
\thanks{M. Zhu is with Wuxi Research Institute of Applied Technologies Tsinghua University, China}
\thanks{W. Ge is with Electrical Engineering Department of Southeast University, China}
}


\maketitle

\begin{abstract}
Defending against distributed denial of service (DDoS) attacks in the Internet is a fundamental problem. However, recent industrial interviews with over 100 security experts from more than ten industry segments indicate that DDoS problems have not been fully addressed. The reasons are twofold. On one hand, many academic proposals that are provably secure witness little real-world deployment. On the other hand, the operation model for existing DDoS-prevention service providers (e.g., Cloudflare, Akamai) is privacy invasive for large organizations (\eg government). In this paper, we present \sys, a new DDoS defense mechanism enabling Internet Service Providers (ISPs) to offer readily deployable and privacy-preserving DDoS prevention services to their customers. At its core, \sys develops a multi-layered defense architecture to defend against a wide spectrum of DDoS attacks. In particular, the flood throttling layer stops amplification-based DDoS attacks; the congestion resolving layer, aiming to prevent sophisticated attacks that cannot be easily filtered, enforces congestion accountability to ensure that legitimate flows are guaranteed to receive their fair shares regardless of attackers' strategies; and finally the user-specific layer allows DDoS victims to enforce self-desired traffic control policies that best satisfy their business requirements. Based on Linux implementation, we demonstrate that \sys is capable to deal with large scale attacks involving millions of attack flows, meanwhile imposing negligible packet processing overhead. Further, our physical testbed experiments and large scale simulations prove that \sys is effective to mitigate various DDoS attacks.
\end{abstract}


\section{Introduction} \label{sec:introduction}
Distributed denial of service (DDoS) attacks have been considered as a serious threat to the availability of Internet. Over the past few decades, both industry and academia make a considerable effort to address this problem. Academia have proposed various approaches, ranging from filtering-based approaches~\cite{practicalIPTrace, advancedIPTrace,AITF,pushback,implementPushback,StopIt}, capability-based approaches~\cite{siff, TVA, netfence, MiddlePolice}, overlay-based systems~\cite{phalanx,sos,mayday}, systems based on future Internet architectures~\cite{scion,aip,xia} and other variance~\cite{speakup,mirage, CDN_on_Demand}. Meanwhile, many large DDoS-protection-as-a-service providers (\eg Akamai, CloudFlare), some of which are unicorns, have played an important role in DDoS prevention. These providers massively over-provision data centers for peak attack traffic loads and then share this capacity across many customers as needed. When under attack, victims use Domain Name System (DNS) or Border Gateway Protocol (BGP) to redirect traffic to the provider rather than their own networks. The DDoS-protection-as-a-service provider applies their proprietary techniques to scrub traffic, separating malicious from benign, and then re-injects only the benign traffic back into the network to be carried to the victim.

Despite such effort, recent industrial interviews with over 100 security engineers from over ten industry segments that are vulnerable to DDoS attacks indicate DDoS attacks have not been fully addressed~\cite{middlepolice-ton}. First, since most of the academic proposals incur significant deployment overhead (\eg requiring software/hardware upgrades from a large number of Autonomous Systems (AS) that are unrelated to the DDoS victim, changing the client network stack such as inserting new packet headers), few of them have ever been deployed in the Internet. Second, existing security-service providers are not cures for DDoS attacks for all types of customer segments. In particular, a prerequisite of using their security services is that a destination site must redirect its network traffic to these service providers. Cloudflare, for instance, will terminate all user Secure Sockets Layer (SSL) connections to the destination at Cloudflare's network edge, and then send back user requests (after applying their secret sauce filtering) to the destination server using new connections. Although this operation model is acceptable for small websites (\eg personal blogs), it is privacy invasive for some large organizations like government, hosting companies and medical foundations.

As a result, these organizations have no other choices but to rely on their Internet Service Providers (ISPs) to block attack traffic. Realizing this issue, in this paper, we propose \sys, a new DDoS defense mechanism focusing on enabling ISPs to offer readily deployable and privacy-preserving DDoS prevention services to their customers. The design of \sys is lessoned from real-world DDoS attacks that intentionally disconnect the victim from the public Internet by overwhelming the victim's inter-connecting links with its ISPs. Thus, \sys proposes to protect the victim by allowing its ISPs to throttle attack traffic, preventing any undesired traffic from reaching the victim. Compared with previous approaches requiring Internet-wide AS cooperation, \sys simply needs \emph{independent deployment} at the victim's direct ISPs to provide immediate DDoS defense. Further, unlike existing security-service providers, an ISP does not need to terminate the victim's connections. Instead, the ISP still operates on \emph{network layer} as usual to completely preserve the victim's application layer privacy. Third, \sys is \emph{lightweight} since it requires no software and hardware upgrades at both the Internet core and clients. Finally, \sys is \emph{performance friendly} because it is overhead-free during normal scenarios by staying completely idle and imposes negligible packet processing overhead during attack mitigation.

\begin{figure*}[t]
  \centering
  \mbox{
    \subfigure[\label{fig:architecture:a}Filtering-based approaches.]{\includegraphics[scale=0.3]{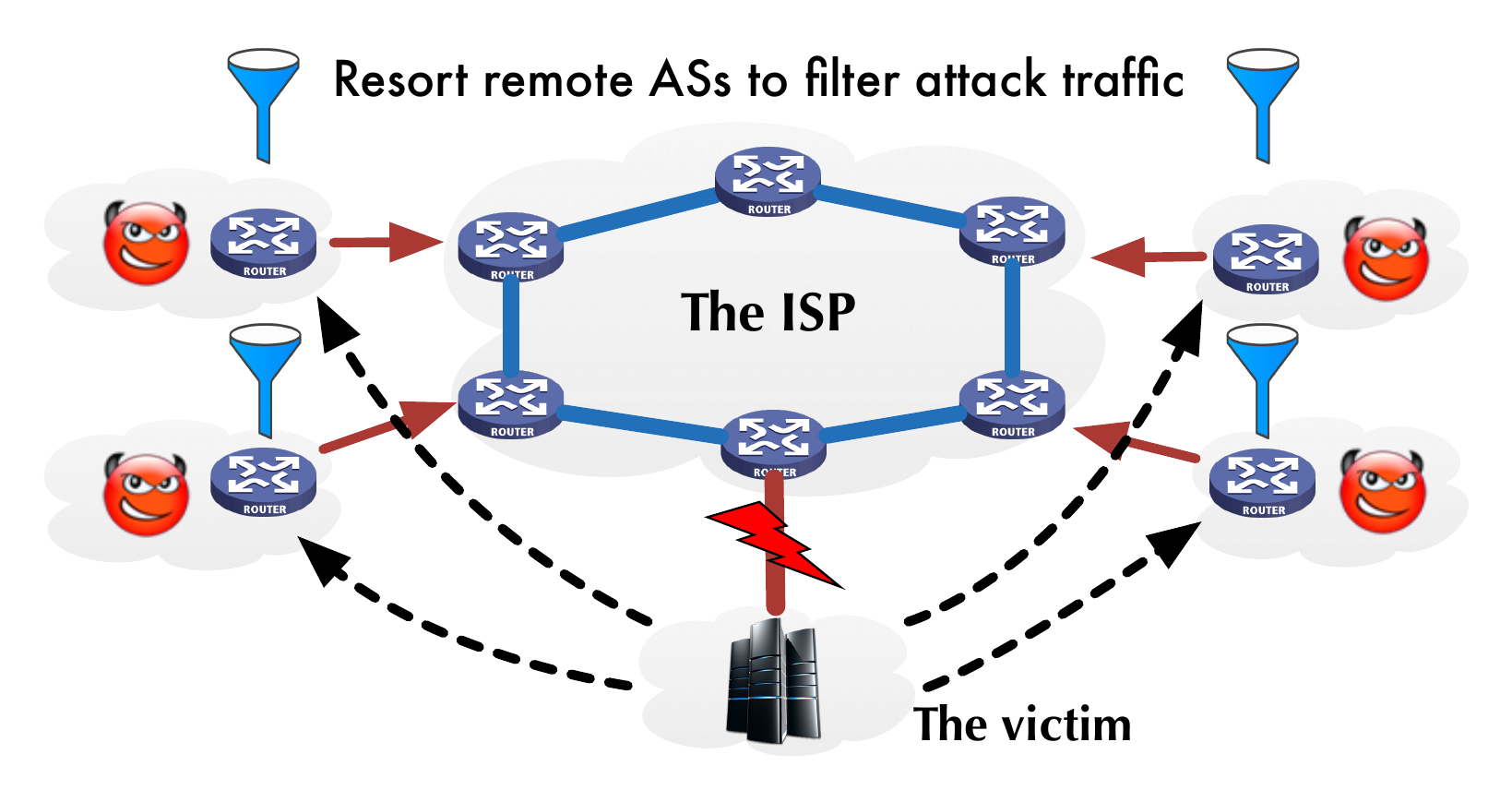}}
    \subfigure[\label{fig:architecture:b}Capability-based approaches.]{\includegraphics[scale=0.3]{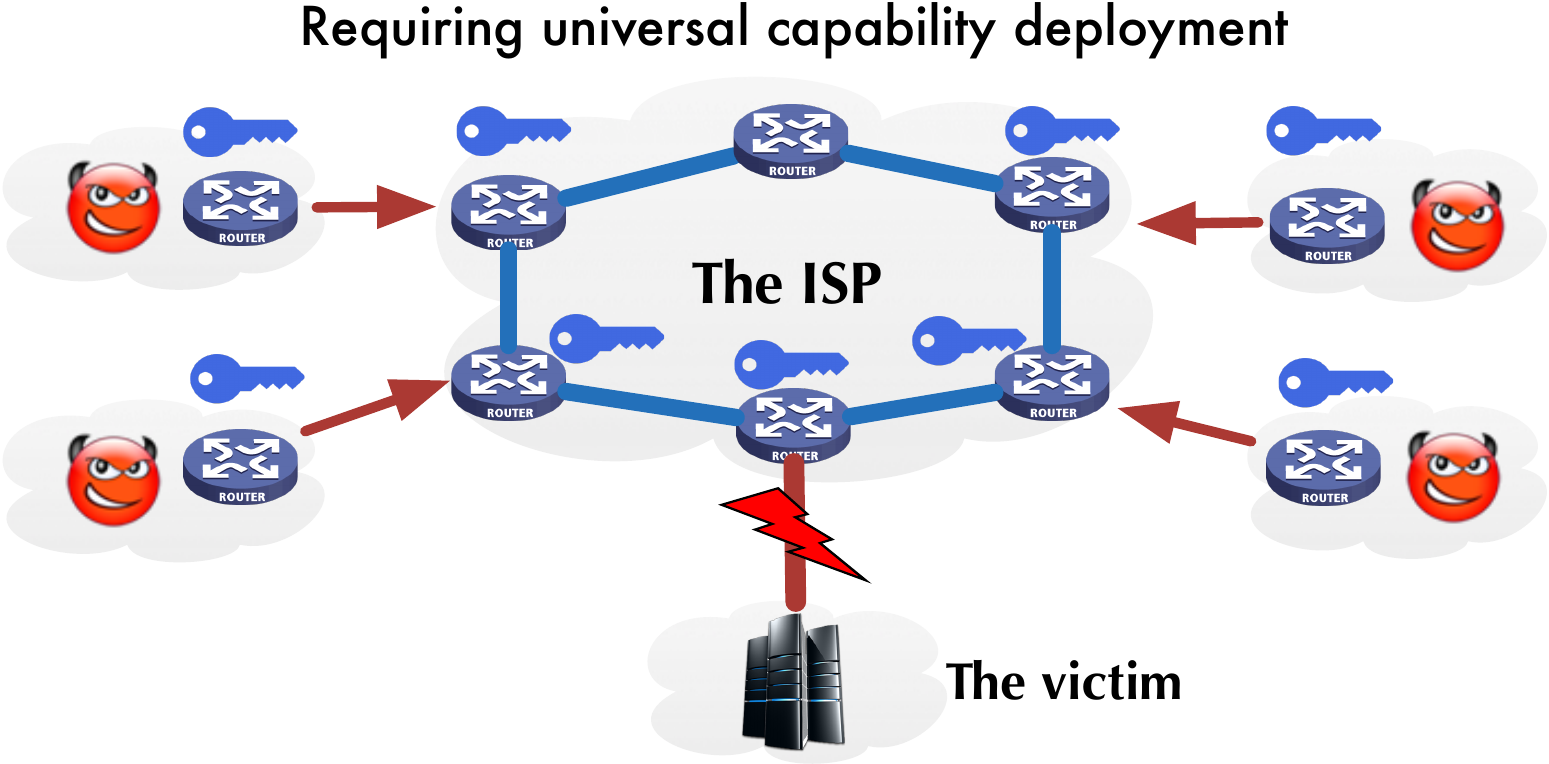}}
     \subfigure[\label{fig:architecture:c}\sys.]{\includegraphics[scale=0.3]{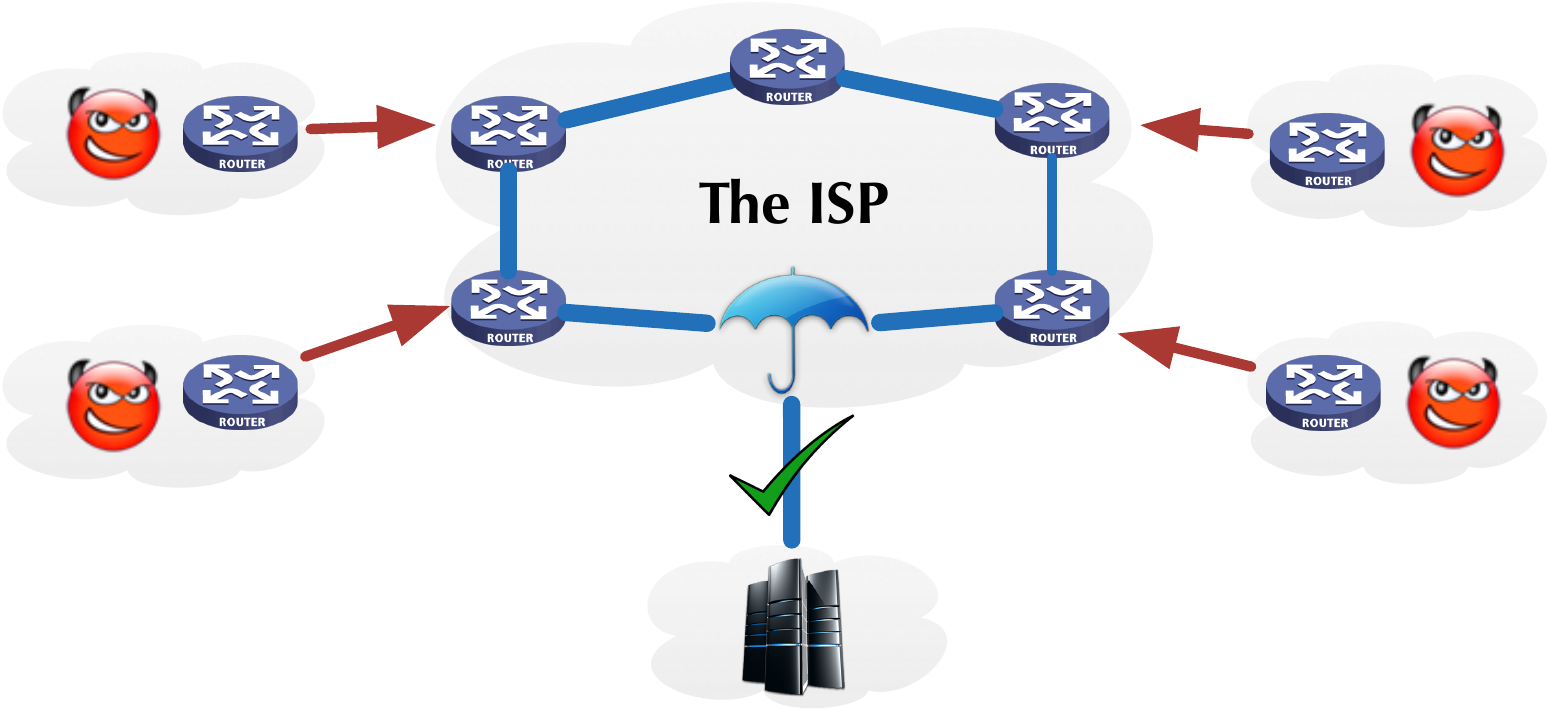}}
    \subfigure[\label{fig:architecture:d}Multi-layered defense.]{\includegraphics[scale=0.3]{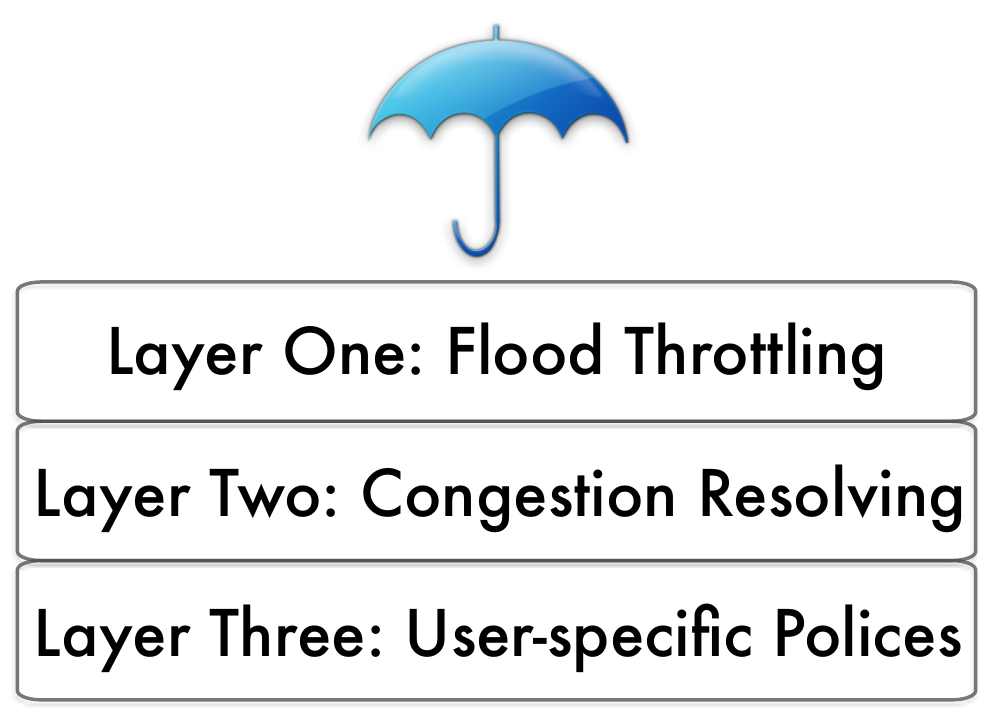}}
    }
  \caption{Part (a) and part (b) show that filtering-based and capability-based approaches are
  difficult to deploy in the Internet since they require Internet-wide AS cooperation and facility upgrades. Part (c)
  shows \sys is immediately deployable at the victim's ISP. Part (d) details \sys's multi-layered defense.}
  \label{fig:architecture}
\end{figure*}

In its design, \sys develops a novel multi-layered defense architecture. In the \emph{flood throttling} layer, \sys defends against the amplification-based attacks that exploit various network protocols (\eg Simple Service Discovery Protocol (SSDP), Network Time Protocol (NTP))~\cite{arbor}. Although such attacks may involve extremely high volume of traffic (\eg hundreds of gigabit per second), they can be effectively detected via static filters and therefore stopped. In the \emph{congestion resolving} layer, \sys defends against more sophisticated attacks in which adversaries may adopt various strategies. \sys brings out a key concept \emph{congestion accountability}~\cite{FlowPolice} to selectively punish users who keep injecting packets in case of severe congestive losses. The congestion resolving layer provides both \emph{guaranteed} and \emph{elastic} bandwidth shares for legitimate flows: \first regardless of attackers' strategies, legitimate users are guaranteed to receive their fair share of the victim's bandwidth; \second when attackers fail to execute their optimal strategy, legitimate clients are able to enjoy more bandwidth shares. The last layer, \emph{user-specific} layer, allows the victim to enforce self-interested traffic policing rules that are most suitable for their business logic. For instance, if the victim never receives certain type of traffic, it can inform \sys to completely block such traffic when  attacks are detected. Similarly, the victim can instruct \sys to reserve bandwidth for premium clients so that they will not be affected by DDoS attacks.

In summary, the major contributions of this paper are the design, implementation and evaluation of \sys, a new DDoS defense approach that enables ISPs to offer readily deployable and privacy-preserving DDoS prevention services to their downstream customers. The novelties of \sys live in the following two dimensions. First, unlike the vast majority of academic DDoS prevention proposals which require extensive Internet core and client network-stack change, \sys only requires lightweight upgrades from business-related entities (\ie the potential DDoS victim itself and its direct ISPs), yielding instant deployability in the current Internet architecture. Second, compared with the existing deployable industrial DDoS mitigation services, \sys, through our novel multi-layer defense architecture, offers both privacy-preserving and  complete DDoS prevention that can deal with a wide spectrum of attacks, and meanwhile offer victim-customizable defense. We implement \sys on Linux to study its scalability and overhead. The result shows that a commodity server can effectively handle $100$ million attackers and introduces merely ${\sim}0.06\mu s$ packet processing overhead. Finally, we perform detailed evaluation on our physical testbed, our flow-level simulator and the ns$3$ packet-level simulator~\cite{ns3} to demonstrate the effectiveness of \sys to mitigate DDoS attacks.

\section{Problem Space and Goals}\label{sec:design_rational}
In this section, we discuss \sys's problem space and its design goals. Completely preventing DDoS attacks is an extremely large scope. The problem space articulates \sys's position within this scope. Further, the design goals of \sys are designed based on the industrial interviews in \cite{middlepolice-ton} so that \sys indeed offers desirable DDoS prevention to those large and privacy-sensitive potential victims such as government and medical infrastructures. We do not however claim that these goals are universally applicable to all types of potential victims (for instance, a web blogger may simply choose CloudFlare to keep her website online).

\subsection{Problem Space}\label{sec:problem_space}
\noindent\textbf{Network-Layer DDoS Mitigation.}
We position \sys as a network-layer DDoS defense approach to stop undesirable traffic from reaching and consuming resources of the victim's network. Specifically, \sys is designed to prevent attackers from exhausting the bandwidth of the inter-network link connecting the victim's network to its ISP so as to keep the victim ``online" even in face of DDoS attacks. Note that \sys should be viewed to be complementary to other solutions addressing DDoS attacks at other layers (\eg layer $7$ HTTP attacks). Only concerted efforts contributed by these solutions can potentially provide complete defense against all types of DDoS attacks.

\noindent\textbf{Adversary Model.}
We consider strong adversaries that can compromise both end-hosts and routers. They are able to launch strategic attacks (\eg the on-off shrew attack~\cite{low-rate}), launch flash attacks with numerous short flows, adopt maliciously tampered transport protocols (\eg poisoned TCP protocols that does not properly adjust rates based on network congestion), leverage the Tor network to hide their identities, spoof addresses and recruit vast amounts of bots to launch large scale DDoS attacks.

\noindent\textbf{Assumptions.}
\sys maintains per-sender state in the congestion resolving layer, and consequently relies on the correctness of source addresses. Such correctness can be assured by the more complete adoption of Ingress Filtering~\cite{ingressFilter, BCP38} or the source authentication schemes~\cite{passport, OPT}. On our way to achieve complete spoof elimination\footnote{In fact, the Spoofer Project~\cite{spoofer_project} extrapolates that only ${\sim}13.6\%$ of addresses are spoofable, indicating a tremendous progress.}, \sys requires victim's additional participation to minimize the chance of source spoofing. In particular, the victim needs to provide a list of authenticated (\ie TCP handshakes with these sources are successfully established) or preferred source IP addresses (based on the victim's routine traffic analysis performed in normal scenarios) so that \sys will only maintain per-sender state for these addresses during attack mitigation.

\subsection{Design Goals}
\noindent\textbf{Readily Deployable.} One major design goal of \sys is to be immediately deployable in the current Internet architecture. To this end, the functionality of \sys relies only on independent deployment at the victim's ISP without further deployment requirements at remote ASs that are unrelated with the victim. As illustrated in Fig. \ref{fig:architecture:c}, \sys can be deployed at the upstream of the link connecting the victim's network and its ISP. In the rest of the paper, we refer to this link as interdomain link  and its bandwidth interdomain bandwidth. Note that \sys deployed at victim's ISP cannot cannot stop DDoS attacks trying to disconnect the victim's ISP from its upstream ISPs. However, the victim's ISP, now becoming a victim itself, should have motivation to protect itself by purchasing \sys's protection from its upstream ISPs. Recursively, DDoS attacks happened at different levels of the Internet hierarchy can be resolved. The neat idea of \sys is that it never requires cooperation among a wide range of (unrelated) ASes. Rather, \emph{independent deployment} is sufficient and effective.

\noindent\textbf{Privacy-Preserving and Customizable DDoS Prevention:} Another primary design goal of \sys is to offer privacy-preserving and customizable DDoS prevention. On one hand, \sys requires no application connection termination at ISPs, allowing them to operate at the network layer as usual, which completely preserves the application-layer privacy of their customers. One the other hand, \sys's multi-layered defense enables ISPs to offer customizable DDoS prevention that is driven by the customer policies.

\noindent\textbf{Lightweight and Performance Friendly:} \sys's deployment is very lightweight: it can be implemented as a software router at the interdomain link, maintaining at most per-source state. Our prototype implementation demonstrates that a commodity server can effectively scale up to deal with millions of states. Further, \sys is completely idle and transparent to applications in normal scenarios, introducing zero overhead. During DDoS attack mitigation, \sys's traffic policing introduces negligible packet processing overhead comparing with previous approaches requiring complicated and expensive operation, such as adding cryptographic capabilities and extra packet headers~\cite{netfence}.

Given these goals, we position \sys as a practical DDoS prevention service, offered by ISPs, that is desirable by large and privacy-sensitive potential DDoS victims.

\section{Design Overview}\label{sec:system}
In its design, \sys develops a three-layered defense architecture to stop undesirable traffic. The user-specific layer, enforcing policies defined by the victim, has priority over the rest two layers, which operate in parallel. \sys is only active when it notices the features of volumetric DDoS attacks against the interdomain link (\eg the link experiences enduring congestion causing severe packet losses). \sys stops traffic policing and becomes idle when the link  restores its normal state. As part of the user-specific layer, the victim is free to define specific rules to determine  when the traffic policing should be initiated or terminated.

\subsection{Flood Throttling Layer}\label{sec:flood_layer}
Flood throttling layer is designed to stop amplification-based attacks, in which attackers send numerous requests, with spoofed source address as the victim's address, to public servers serving certain Internet protocols (\eg Network Time Protocol, Domain Name System). As a result, the victim receives extremely high volume of responses, resulting interdomain bandwidth exhaustion and disconnection from its ISP. Although the attack volumes can as high as around 600Gpbs~\cite{arbor}, these attacks are easy to be detected. According to the analysis in \cite{arbor}, only seven network service protocols are exploited to launch amplification-based DDoS attacks and over 87\% of them rely on UDP with specific port numbers. Thus, by installing a set of static filters based on these network service protocols, \sys can effectively throttle these large-yet-easy-to-catch DDoS attacks.

In the case where the victim does receive traffic for certain network service protocols that may be exploited to launch amplification attacks, \sys can leverage the Weighted Fair Queuing technique to minimize the potential affect of these attacks. For instance, if based on traffic analysis in normal scenarios, $10\%$ of the victim's traffic traversing the interdomain link is NTP (on UDP port 123), then such traffic should be served in a queue whose normalized weight is configured as 0.1. The weighted fair queuing ensures that the victim always has sufficient bandwidth to process other flows (which are served in separate queues), regardless of how much NTP traffic is thrown to the victim. 

\subsection{Congestion Resolving Layer}\label{sec:congestion_layer}
The congestion resolving layer is designed to stop subtle and sophisticated DDoS attacks that rely on numerous seemly legitimated TCP traffic. The crucial part of the defense is to enforce \emph{congestion accountability}~\cite{FlowPolice} so as to punish attackers who keep injecting vast amounts of traffic in face of congestive losses. Specifically, in volumetric DDoS attacks, overloaded routers drop packets from all users regardless of which users cause the \emph{enduring} congestion. In other words, congestion accountability is not considered while dropping packets. Consequently, legitimate users that run congestion-friendly protocol (\eg TCP) are falsely penalized during the congestion since it is actually caused by attackers. By enforcing congestion accountability, \sys is able to selectively punish misbehaved flows and therefore stop attack traffic.

\sys analyzes each user's congestion accountability from the perspective of the goal of network usage. Legitimate users aim to deliver or receive data via the network. Thus, when encountering congestion, a legitimate TCP sender tries to \emph{relieve} the congestion by reducing its rate because its receiver cannot decode the data if some packets are lost. Sending more packets therefore makes no progress on finishing the data transfer, but it causes more severe congestion. On the contrary, attackers focus on exhausting network resources and care little about delivering data. As a result, they consistently generate traffic to \emph{contribute to} the congestion regardless of how many packets have been dropped. Therefore, users who overlook packet losses and continuously inject packets are accountable for the enduring congestion happened during DDoS attacks.

To resolve the congestion, \sys keeps a \emph{rate limiting window} for each user to prevent any user from sending faster than its rate limiting window. The size of the rate limiting window is determined by the rate limiting algorithm (\S \ref{sec:rateLimit}), taking input as the users' sending rates and packet losses. All information needed to make rate limiting decisions is recorded in \sys's flow table (\S \ref{sec:flowTable}). The design of the rate limiting algorithm ensures \first the more aggressively attackers behave, the less bandwidth shares they will obtain; \second each legitimate client is guaranteed to receive the per-sender fair share of the congestion resolving layer's bandwidth regardless of attackers' strategies (note that part of the interdomain bandwidth may be used in other layers). Further, legitimate users may obtain more bandwidth than the per-sender fair share when attackers fail to execute their optimal strategy.


\subsection{The User-specific Layer}\label{sec:user_layer}
The goal of adding the user-specific defense layer is to provide the flexibility for the victim to enforce self-interested traffic control policies that are most suitable for their business logic, including adopting different fairness metrics from the \sys's default one (per-sender fairness) and offering proactive DDoS defense for premium clients so that they will never be disconnected from the victim.
Allowing user-specific policies differs \sys from previous in-network DDoS prevention mechanisms that force the victim to accept the single policy proposed by these approaches. Thus \sys creates extra deployment incentives for ISPs by enabling them to offer customized DDoS defense to consumers.

\section{Design Details}
Since the flood throttling layer is straightforward in its design and the user-specific layer is typically driven by the victim, we focus on elaborating the congestion resolving layer in this section. However, we have a full implementation of \sys with all three layers in \S \ref{sec:implementation}.

\subsection{Flow Table}\label{sec:flowTable}
\sys's flow table maintains per-sender network usage. Specifically, all packets sent from the same source are aggregated (and defined) as one \emph{coflow}\footnote{The concept of coflow is also introduced in data centers, meaning a group of flows from the same task~\cite{coflow2}.} and the flow table maintains states for each coflow. As discussed in the \emph{Assumptions} section of \S \ref{sec:problem_space}, the flow table maintains state only for the set of source IP addresses explicitly provided by the victim to prevent adversaries from exhausting table state via spoofed addresses. \sys does not keep states for each individual TCP flow (identified by its $5$-tuple) since the behavior of one single flow may not reflect the intention of the sender (malicious or not). For instance, one bot keeps sending new flows to the victim although previous flows experience severe losses. Even each individual flow may be a legitimate TCP flow, the bot is actually acting maliciously. However, if we interpret its behaviors from the coflow's perspective, we can figure out that the bot continuously creates traffic in face of congestive losses. Thus it is accountable for the congestion and will be rate limited. In the rest of the paper, unless otherwise stated, flow and coflow are used interchangeably.

\begin{figure}[t]
\centering
\begin{tabular}{|>{}c|>{}c|>{}c|>{}c|>{}c|>{}c|}
\hline
    $f$&
    $\mathcal{T}_A$&
    $\mathcal{W}_R$&
    $\mathcal{P}_R$&
    $\mathcal{P}_D$&
    $\mathcal{L}_R$\\
\hline
	$64$&
	$32$&
	$32$&
	$32$&
	$32$&
	$64$\\
\hline
\end{tabular}
\caption{Each field in a flow entry and its corresponding size (bits).}
\label{fig:flow_entry}
\end{figure}

Each flow entry (identified by its source address $f$) in the flow table is composed of
a timestamp $\mathcal{T}_A$, $f$'s rate limiting window $\mathcal{W}_R$,
the number of packets $\mathcal{P}_R$ received from $f$, the number
of dropped packets $\mathcal{P}_D$ from $f$ and its packet loss rate $\mathcal{L}_R$.
Further, \sys maintains $\mathcal{W}_R^T$, the sum of rate limiting
windows of all flows, shared by all flow entries.
These information is necessary for the rate limiting algorithm.

\subsection{Rate Limiting Algorithm}\label{sec:rateLimit}
The rate limiting algorithm is designed to enforce congestion accountability by punishing misbehaved users who
keep sending packets in face of severe congestive losses. By early dropping the
undesirable packets, \sys can effectively prevent bandwidth exhaustion.
In its design, the algorithm executes periodic rate limiting for each flow during DDoS attacks.
Specifically, in each \emph{detection period}, the number of packets allowed for
each flow (or sender) is limited by its rate limiting window $\mathcal{W}_R$.
The $\mathcal{W}_R$ is updated every detection period according
to the flow's information recorded in the flow table, such as the flow's packet loss rate $\mathcal{L}_R$ and
its transmission rate $\mathcal{P}_R$.

\subsubsection{Populating the flow table}
Assume at time $ts$, a new flow $f$ is initiated.
\sys creates a flow entry for $f$ in its flow table.
All fields of the entry are initialized to be zero.
Then \sys updates $\mathcal{T}_A$ as $ts$, increases $\mathcal{P}_R$ by one
and sets the initial $\mathcal{W}_R$ as the pre-defined fair share rate
$\mathcal{W}_{fair}$ (discussed in \S\ref{sec:paraSettings}).
From then on, \sys increases $\mathcal{P}_R$ by one for each arrived
packet of $f$ until the end of the current detection period (\eg the end of the first detection period).
\sys uses packet arrival time to detect when it should start a new detection period for $f$.
Specifically, letting $\mathcal{D}_p$ denote the length of detection period,
when received a packet with arrival time $t_0 {>} \mathcal{T}_A {+} \mathcal{D}_p$,
\sys realizes that this packet is the first one received in the new detection period.
Then \sys performs the following updates in order:
\first Set $\mathcal{T}_A=t_0$;
\second Update $\mathcal{W}_R$ and $\mathcal{L}_R$ according to the Algorithm \ref{algo:rateLimit};
\third Reset $\mathcal{P}_R$ and $\mathcal{P}_D$ as zero.

\subsubsection{The rate limiting algorithm}
\newcommand{\algrule}[1][.2pt]{\par\vskip.5\baselineskip\hrule height #1\par\vskip.5\baselineskip}
\begin{algorithm}[t]
\SetDataSty{textsl}
\SetKwData{CQ}{CongestionQueue}
\SetKwData{RL}{RateLimitingDecision}
\SetKwData{FC}{\textbf{Function:}}
\SetKwData{IP}{\textbf{Input:}}
\SetKwData{OP}{\textbf{Output:}}
\SetKwData{Main}{\textbf{Main Procedure:}}
\SetKwData{WS}{RateLimitingWindow}
\SetKwData{IN}{\textbf{Flow Initialization:}}

\footnotesize

\IP \\
\quad The service queue $\mathcal{Q}_C$ executing the FIFO principle.\\
\quad $f$'s entry in the flow table. \\
\quad $\mathcal{B}$: the bandwidth available in the congestion resolving layer. \\
\quad $\mathcal{W}_R^T$: the sum of all flows' rate liming windows. \\
\quad System related parameters: the detection period length $\mathcal{D}_p$, the
weight $\lambda$ for previous packet losses, the packet loss rate threshold $L_{Th}$ and
the per-sender fair share rate $\mathcal{W}_{fair}$.\\
\OP \\
\quad Updated $\mathcal{Q}_C$,  $\mathcal{W}_R^T$ and the flow entry of $f$.

\algrule[0.5pt]
\IN \\
\quad $\mathcal{W}_R = \mathcal{W}_{fair}$~\;

\algrule[0.5pt]
\Main \\

	\For{each arrived packet $\mathcal{P}$ of flow $f$}{
		$\mathcal{P}_R \leftarrow \mathcal{P}_R + 1$~\;
				\If{$\mathcal{P}_R > \mathcal{W}_R$} {
					Drop $\mathcal{P}$~; ~~$\mathcal{P}_D \leftarrow \mathcal{P}_D + 1$~\;
				}
				\lElse{
					\CQ{$\mathcal{P}$}
				}
			
			\If{$\mathcal{P}$ is the first packet in a new detection period} {
				\RL{$\mathcal{\mathcal{W}_R, \mathcal{L}_R, \mathcal{P}_R}, \mathcal{P}_D$}\;
			}

	}

\algrule[0.5pt]

\textbf{\FC} \RL{$\mathcal{W}_R, \mathcal{L}_R, \mathcal{P}_R, \mathcal{P}_D$}: \\
\quad $recentLoss \leftarrow \mathcal{P}_D / \mathcal{P}_R$~\;
\quad $packetLoss \leftarrow \lambda \cdot \mathcal{L}_R + (1-\lambda) \cdot recentLoss$~\;	
\quad $\mathcal{W}_R^{original} \leftarrow  \mathcal{W}_R$~; ~~$\mathcal{L}_R \leftarrow packetLoss$~\;
\quad \lIf{$packetLoss > L_{Th}$ and $\mathcal{P}_R > \mathcal{W}_{fair}$}{
				$\mathcal{W}_R \leftarrow \mathcal{W}_R / 2$
			}
\quad	\lElse {
				$\mathcal{W}_R \leftarrow$ \WS{$\mathcal{W}_R$}
			}
\quad $\mathcal{W}_R^T \leftarrow \mathcal{W}_R^T + \mathcal{W}_R - \mathcal{W}_R^{original}$~\;

\algrule[0.5pt]

\textbf{\FC} \CQ{$\mathcal{P}$}: \\
\quad \If{the queue $\mathcal{Q}_C$ is full}{
	Drop $\mathcal{P}$~\;
	$\mathcal{P}_D \leftarrow \mathcal{P}_D + 1$~\;
}
\quad \lElse{
	Append $\mathcal{P}$ to $\mathcal{Q}_C$
}

\algrule[0.5pt]

\textbf{\FC} \WS{$\mathcal{W}_R$}: \\
\quad \textbf{return} $\frac{\mathcal{W}_R}{\mathcal{W}_R^T}\cdot \mathcal{B}$~\;

\caption{\bf Rate Limiting Algorithm}\label{algo:rateLimit}
\end{algorithm}

\normalsize

At the very high level, the rate limiting algorithm determines the allowed rate
for each flow based on its congestion accountability. In particular,
the rate limiting windows of congestion-accountable flows (with both high packet loss rates and
high transmission rates) are significantly reduced.
Flows respecting packet losses by adjusting sending rates accordingly are
guaranteed to receive per-sender fair share of the bandwidth.
We adopt such a fairness metric because it is the optimal one that can be
guaranteed for legitimate users under strategic attacks. The proof is straightforward:
by behaving in the exact same way as legitimate users,
attackers can receive at least per-sender fair share, meaning that the
optimal guaranteed share for a legitimate user is also the per-sender fair share.
However, the algorithm allows legitimate users to obtain more bandwidth
shares when attackers fail to execute their optimal strategy.

\sys performs periodic rate limiting. In each detection period,
\sys learns each flow's transmission rate and packet loss rate to determine its $\mathcal{W}_R$.
One flow $f$'s transmission rate is quantified by $\mathcal{P}_R$,
the number of received packets from $f$ in the current period.
$f$'s packets may be dropped for two reasons: \first $f$'s
sending rate exceeds its $\mathcal{W}_R$ or \second the service queue is full due to congestion.
$f$'s packet loss rate $\mathcal{L}_R$ in the current period is the ratio of dropped packets
to received packets. While making rate limiting decisions,
\sys adopts the metric $packetLoss$, which incorporates both packet losses in
the current period and previous packet losses. Such a design prevents
attackers from hiding their previous packet losses by stopping transmitting for a while before
sending a new traffic burst (\eg the on-off shrew attack~\cite{low-rate}).
If both $packetLoss$ and $\mathcal{P}_R$ exceed their pre-defined thresholds,
\sys classifies $f$ as a maliciously behaved flow and reduces
its $\mathcal{W}_R$ by half.

We explain two design details of the rate limiting algorithm.
To begin with, the algorithm cannot make the rate limiting decision for a fresh flow
in its first detection period since \sys has not learned its packet loss rate and sending rate yet.
Thus \sys initializes its $\mathcal{W}_R$ as the pre-defined
per-sender fair share rate $\mathcal{W}_{fair}$ in the first detection period,
preventing attackers from exhausting bandwidth by creating new flows.
Besides $\mathcal{W}_{fair}$, the algorithm relies on anther three
system related parameters: $\mathcal{D}_p$, $\lambda$ and $L_{Th}$.
We discuss the reasoning for parameter settings in \S\ref{sec:paraSettings}.
Further, the \textsl{RateLimitingWindow} function returns the allowed bandwidth for $f$.
We need to convert the bandwidth value into the number of $1.5$KB packets allowed in
one detection period, which will be $f$'s updated $\mathcal{W}_R$.

We close our algorithm design with the remark concerning the SYN flooding attack.
When a SYN packet's source address is matched by one flow entry (meaning
the source address has been authenticated), it will be treated in the same way as
regular packets from the source. Thus sending SYN packets also consumes attackers'
bandwidth budget. SYN packets with unverified sources are appended to a
queue with bounded bandwidth (\eg $5\%$ of $\mathcal{B}$).
Thus the spoofed SYN flooding cannot compromise \sys's defense. Regular packets with
unidentifiable sources in the flow table are denied.

\vspace*{-0.15in}
\subsection{Parameter Settings}\label{sec:paraSettings}
\noindent{$\bm{\mathcal{D}_p}$:}
The length of detection period should be long enough for \sys to characterize each flow's behaviors during the congestion as so to determine its congestion accountability. In particular, $\mathcal{D}_p$ needs to be long enough to allow legitimate users to adapt to the congestion so as to maintain a very low packet loss rate. Meanwhile, \sys is confident that users with high packet loss rates during such a long period of
time are misbehaving. Given that TCP adjusts its window every RTT, $\mathcal{D}_p$ should be much longer than typical Internet RTTs (hundreds of milliseconds based on the CAIDA's measurement~\cite{RTT_measure}). If $\mathcal{D}_p$ is too short, the legitimate flows may fail to adapt to the congestion quickly enough, resulting in inaccurate and highly fluctuating loss rates for them. On the contrary, $\mathcal{D}_p$ cannot be too long to avoid slow reaction to attacks. Balancing the two factors, $2{\sim}6$ seconds are reasonable choices for $\mathcal{D}_p$.

\noindent{$\bm{\lambda}$}: The value of $\lambda$
represents the weight assigned to one flow's previous packet losses.
To defend against the on-off shrew attack~\cite{low-rate}, \sys gives a non-trivial weight
to previous packet losses by setting $\lambda=0.5$. Therefore, once a flow misbehaves, it will have a bad reputation for a while. In order to regain reputation, the flow would have to honor congestion by reducing its sending rate when experienced packet losses.

\noindent{$\bm{L_{Th}}$:} The value of $L_{Th}$ should be larger than normal
packet loss rates to avoid false positives.
According to the previous measurements~\cite{tcpMeasure, internetMeasure},
we set $L_{Th} = 5\%$.

\noindent{$\bm{\mathcal{W}_{fair}}$:} We define the fair share of each flow
as $\mathcal{W}_{fair} {=} \mathcal{B}/\mathcal{N}$, where
$\mathcal{N}$ is the number of flows in the flow table.\footnote{When
\sys is activated from the idle state, $\mathcal{N}$ can be obtained from the network monitoring and logging
tools such as the NetFlow~\cite{netflow}.} Again, the bandwidth
value needs to be converted into the number of packets.
$\mathcal{W}_{fair}$ is updated when new flows are initiated.
As we aggregate all traffic from the same sender as one flow,
$\mathcal{W}_{fair}$ may be updated less frequently
than each flow's $\mathcal{W}_R$.

\subsection{Algorithm Analysis}\label{sec:algo_analysis}
In this section, we prove that the rate limiting algorithm provides both
guaranteed and elastic bandwidth shares for legitimate users:
they are guaranteed to obtain the per-sender fair share and can potentially
obtain more bandwidth shares. We first state the optimal
bandwidth shares attackers can get.

\begin{lemma}\label{lemma:attack}
Given that $\mathcal{N}_L$ legitimate flows and $\mathcal{N}_A$ attack flows share the
congestion resolving layer's bandwidth $\mathcal{B}$, regardless of attackers' strategies,
the aggregated bandwidth that can be obtained by attack flows is \textbf{at most} $\frac{(1+L_{Th})\cdot \mathcal{N}_A \cdot
\mathcal{B}}{\mathcal{N}_L + \mathcal{N}_A}$.
\end{lemma}
\begin{proof}
\sys initializes each flow's $\mathcal{W}_R$ as per-sender fair share rate $\mathcal{W}_{fair}$.
Thus attackers can obtain $\frac{\mathcal{N}_A \cdot \mathcal{B}}{\mathcal{N}_L + \mathcal{N}_A}$ initial bandwidth.
The rate limiting algorithm allows a maximum $L_{Th}$ loss rate before
further reducing one flow's rate limiting window. Thus the optimal strategy
for an attack flow is to strictly comply with \sys's rate limiting by
sending no more than $1+L_{Th}$ times its rate limiting window.
Otherwise, its bandwidth share will be further reduced.
In a hypothetical situation where attackers are able to know their exact rate limiters
and control their packet losses remotely, they can obtain at most
$\frac{(1+L_{Th})\cdot \mathcal{N}_A \cdot \mathcal{B}}{\mathcal{N}_L + \mathcal{N}_A}$.
\end{proof}

Based on the Lemma \ref{lemma:attack}, we obtain the following theorem.

\begin{theorem}\label{theorem:legitimate}
Each legitimate flow can obtain \textbf{at least} $\frac{(1+L_{Th})\cdot\mathcal{B}}{\mathcal{N}_L + \mathcal{N}_A}$
bandwidth share, given that its transport protocol can
fully utilize the allowed bandwidth.
\end{theorem}
\begin{proof}
As each legitimate flow complies with \sys's rate limiting,
it is guaranteed to receive the per-sender fair share.
However, the per-sender fair share
is the lower-bound of its bandwidth share. When attackers fail to adopt
their optimal strategy (\eg sending flat rates),
their rate limiting windows are significantly reduced.
As a result, legitimate flows' windows, returned by the \textsl{RateLimitingWindow} function,
will be increased since $\mathcal{W}_R^T$ is reduced. Thus legitimate
flows can receive more bandwidth than the per-sender fair share.
\end{proof}

To sum up, Lemma~\ref{lemma:attack} and Theorem~\ref{theorem:legitimate} spell
a dilemma for attackers: sending aggressively
rapidly throttles themselves to almost zero bandwidth share whereas
complying with \sys's rate limiting makes their attacks in vain.

\section{Implementation and Evaluation}\label{sec:implementation}
In this section, we describe the implementation and evaluation of \sys. We
first demonstrate that \sys is scalable to deal with DDoS attacks
involving millions of attack flows and meanwhile introduces negligible packet processing overhead.
Then we implement all three layers of \sys's defense on our physical testbed to
evaluate \sys's performance. Further, we add detailed simulations to prove that
\sys is effective to mitigate large scale DDoS attacks.

\subsection{Overhead and Scalability Analysis}\label{sec:scalability}
The flood throttling layer can be implemented as weighted fair queuing. Thus
it introduces almost zero overhead since \sys does not maintain any extra states.
The overhead of user-specific layer depends on specific policies.
To learn the overhead of \sys's congestion resolving layer
(\eg per-packet processing overhead and memory consumption),
we implement \sys's rate limiting logic on a Dell
PowerEdge R320 server shipped with an $8$-core Intel E$5$-$1410$
$2.8$GHz CPU and $24$GB memory.
As illustrated in Fig. \ref{fig:flow_entry}, the total size of a single flow
entry is $24$ bytes. Thus, even when \sys maintains a
flow table with $100$ million flows, the memory consumption is just a few gigabytes,
which can be easily supported by commodity servers.
We show both the memory consumption and per-packet processing overhead
for three table sizes ($1$, $10$ and $100$ million entries)
in Figure \ref{fig:scalability}.

\begin{figure}[t]
  \centering
  \mbox{
    \subfigure{\includegraphics[scale=0.45]{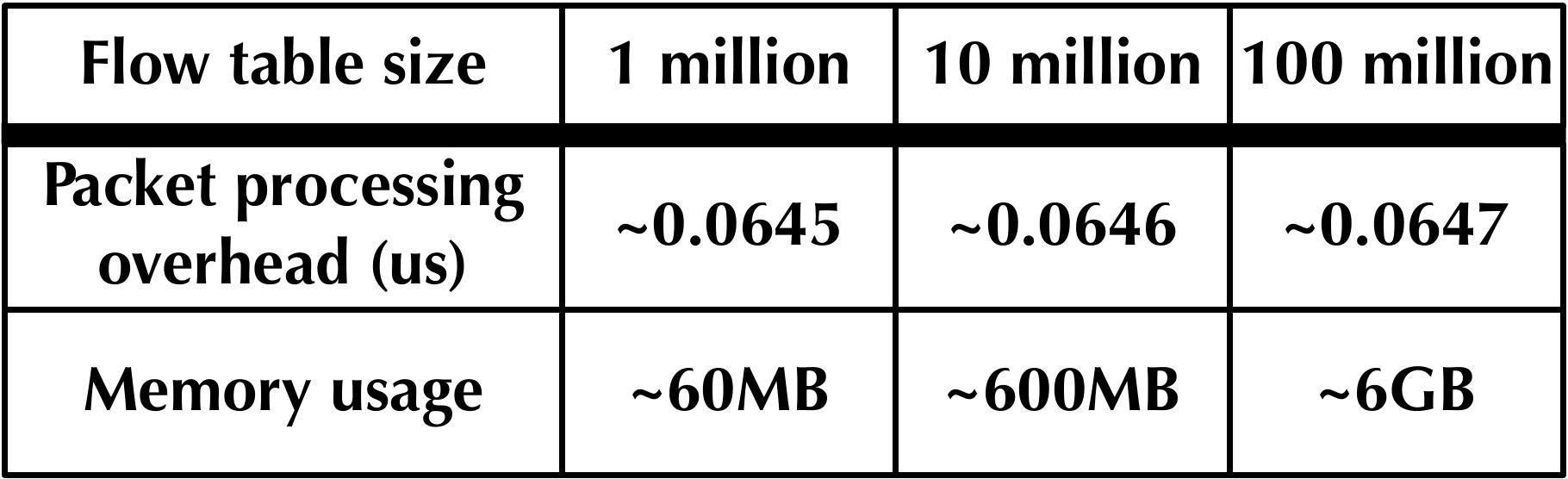}}
    }
  \caption{\sys's memory consumption and per-packet processing overhead.}
  \label{fig:scalability}
\end{figure}

For the largest table size, the memory consumption is around $6$GB,\footnote{Note that
the memory usage for $100$ million flow entries is not exactly $2.4$GB since
we adopt the \texttt{map} data structure to implement the flow table, resulting
in additional memory consumption.} indicating that
memory will not become the bottleneck of \sys's implementation.
Further, the per-packet processing overhead remains almost the same when the number of
flow entries increases from $1$ million to $100$ million. Thus \sys
can effectively scale up to deal with DDoS attacks involving millions of attack flows.
Moreover, the ${\sim}0.06\mu s$ per-packet processing overhead
is negligible even for the $10$Gbps Ethernet with around $1.2\mu s$ per-packet processing time.
Thus the victim can still enjoy high speed Ethernet after deploying \sys.
Note that the implementation of \sys's rate limiting algorithm may be optimized according to the
system hardware to further reduce the overhead.

\subsection{Testbed Experiments}\label{sec:testbed}

We implement a prototype of \sys on our testbed consisting $9$ servers, illustrated in Fig. \ref{fig:testbed_arch}.
Each server is the same Dell PowerEdge R$320$ used to learn \sys's overhead (\S\ref{sec:scalability}).
The server is running Debian $6.0$-$64$bit with Linux $2.6.38.3$ kernel and is installed a
Broadcom BCM$5719$ NetXtreme Gigabit Ethernet NIC.
We organize the servers into $7$ senders (either attackers or legitimate users),
one software router implementing \sys's three-layered defense and one victim, as illustrated in Fig. \ref{fig:testbed_arch:b}.
Thus the interdomain bandwidth in the testbed is $1$Gbps.

The flood throttling layer is implemented as a weighted fair queuing module at the
output port of the software router. The module serves TCP flows and UDP flows
in two separate queues with different weights. The normalized weight for
TCP flows' queue is $0.9$ whereas UDP flows' queue weight is $0.1$ (again the
victim can overwrite the setting).
Each queue has its own dedicated buffer since UDP flows will consume
all buffers when competing with TCP flows, resulting in almost zero throughput for TCP traffic.
With the protection of the flood throttling layer, TCP flows with
sufficient traffic demand can obtain $900$Mbps share of the interdomain link,
regardless of how much UDP traffic is thrown to the victim.

The congestion resolving layer is composed of a set of rate limiters.
Each rate limiter is implemented via the Hierarchical Token Bucket (HTB) of
the Linux's Traffic Control~\cite{linux_tc}. The bandwidth of each rate limiter
is each flow's $\mathcal{W}_R$, determined by \sys's rate limiting algorithm,
to ensure no flow can send faster than its $\mathcal{W}_R$. We set $\mathcal{D}_p{=}5$s
in the implementation.

\begin{figure}[t]
  \centering
  \mbox{
    \subfigure[\label{fig:testbed_arch:a}Servers.]{\includegraphics[scale=0.3]{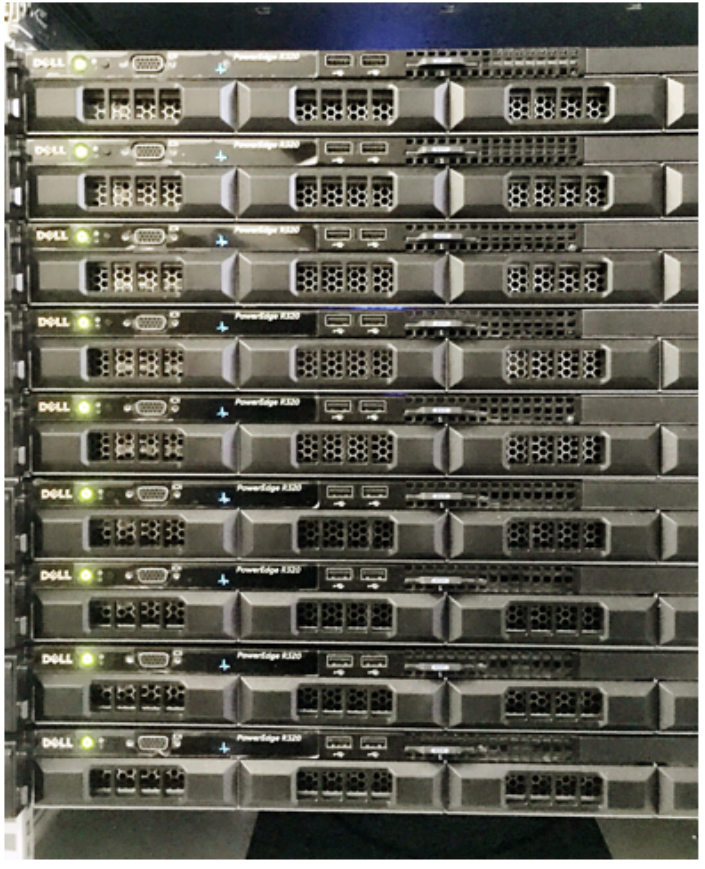}}\quad\quad
    \subfigure[\label{fig:testbed_arch:b}Testbed topology.]{\includegraphics[scale=0.28]{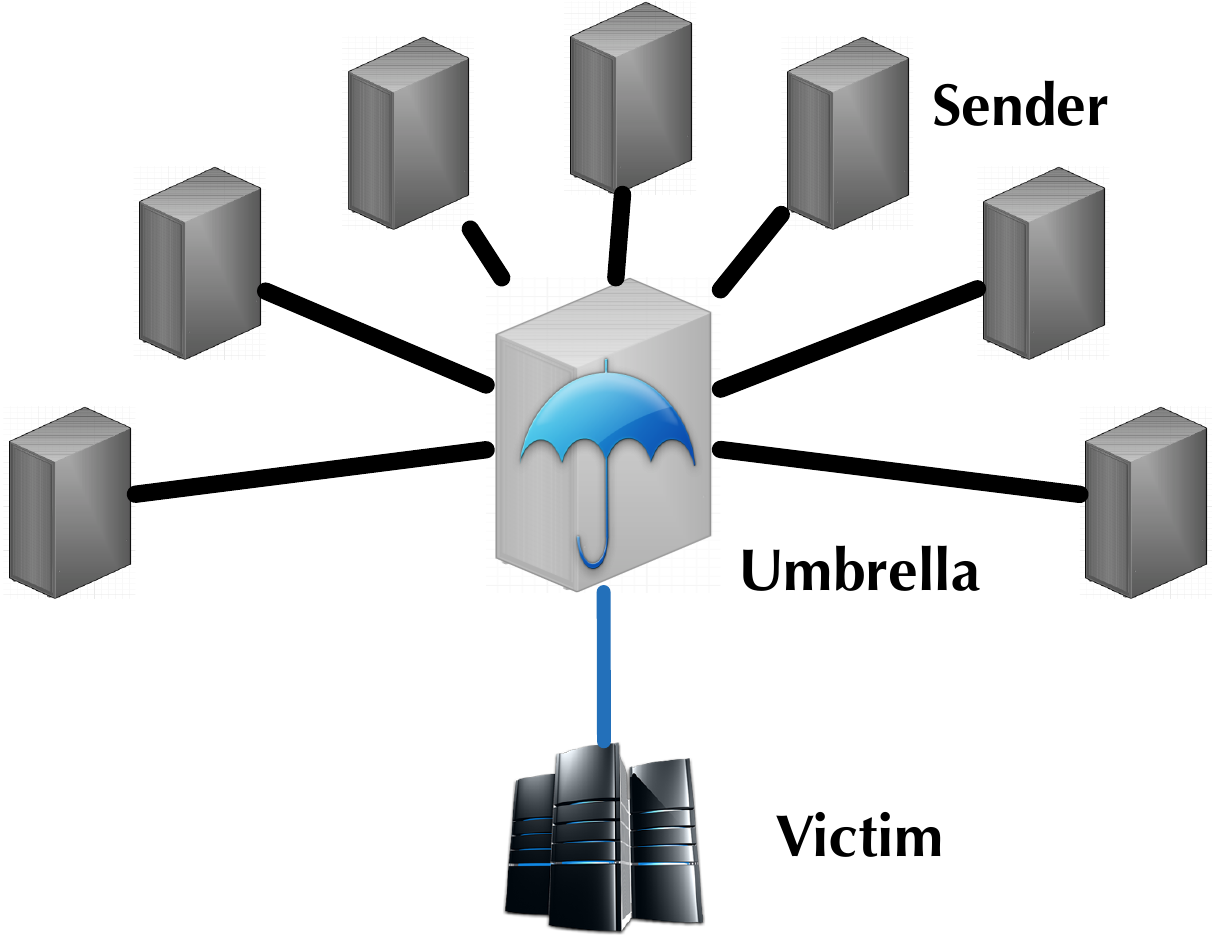}}
  }
  \caption{The prototype of \sys on our testbed.}
  \label{fig:testbed_arch}
\end{figure}

\begin{figure*}[t]
  \centering
  \mbox{
    \subfigure[\label{fig:testbed:a}Layer one defense.]{\includegraphics[scale=0.35]{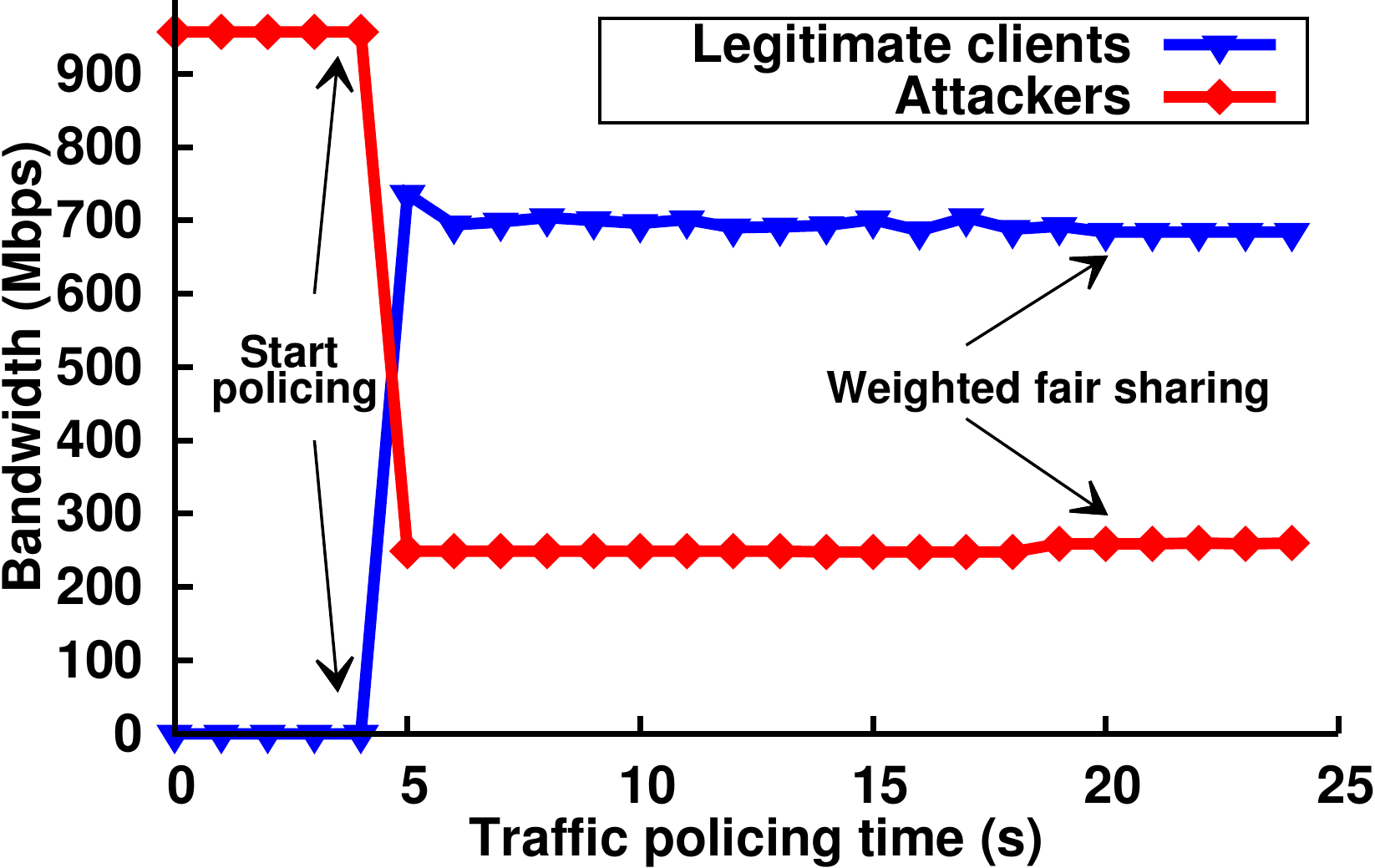}}
    \subfigure[\label{fig:testbed:b}Layer two defense.]{\includegraphics[scale=0.35]{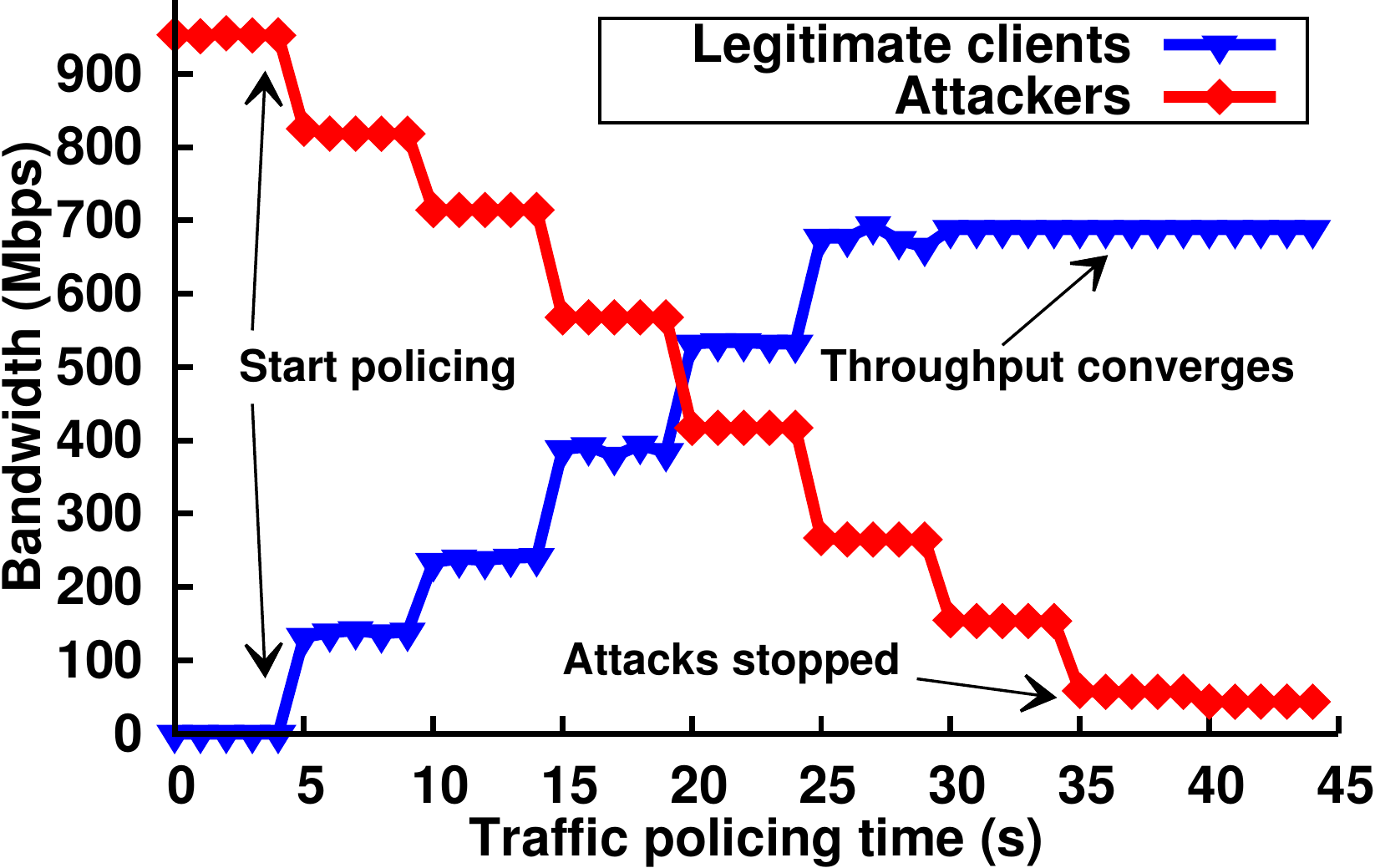}}
     \subfigure[\label{fig:testbed:c}Layer two and three defense.]{\includegraphics[scale=0.35]{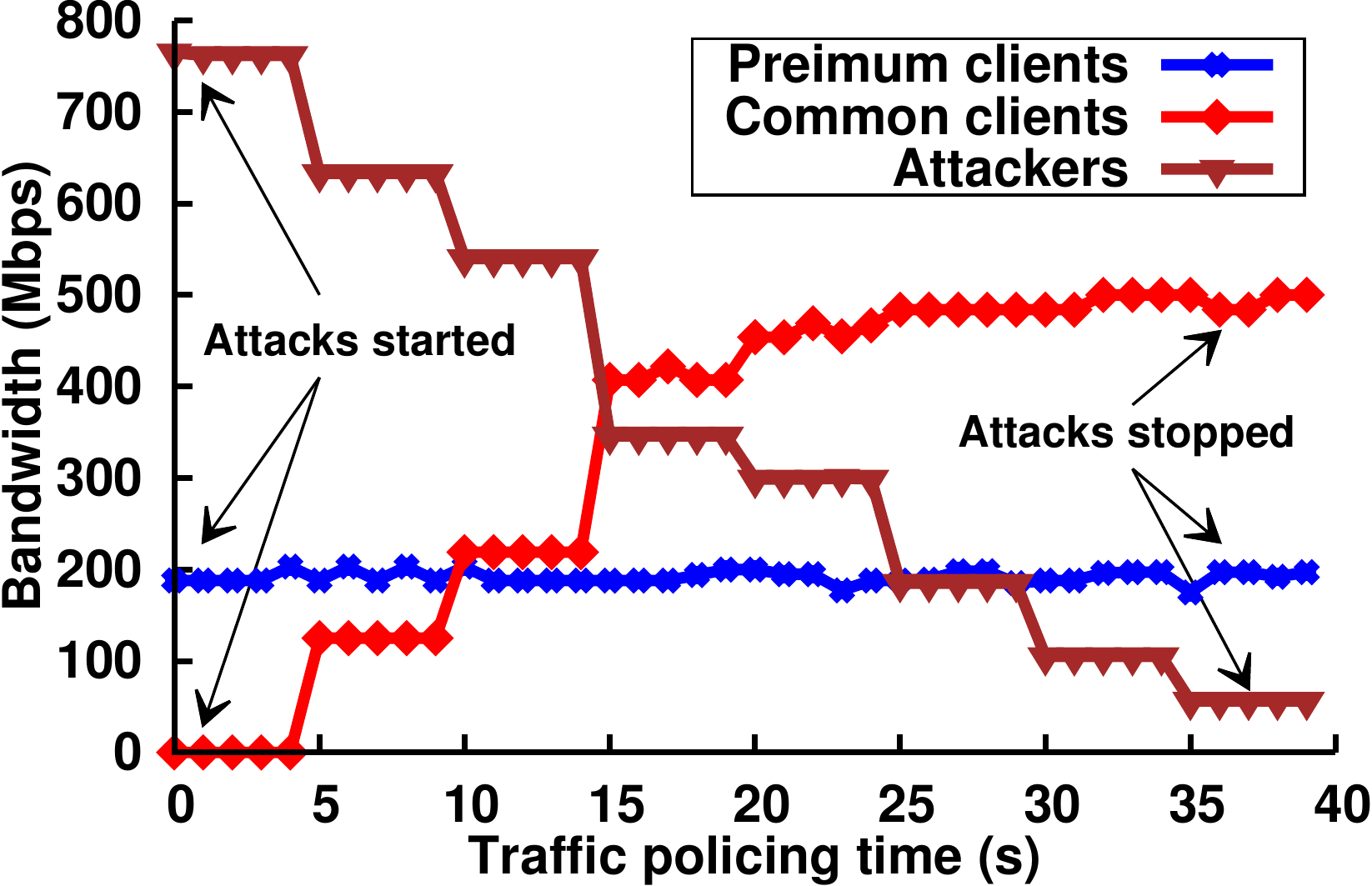}}
    }
  \caption{Testbed experiments.}
  \label{fig:testbed}
\end{figure*}

\begin{figure*}[t]
  \centering
  \mbox{
    \subfigure[\label{fig:simulation:a}On-off shrew attacks.]{\includegraphics[scale=0.35]{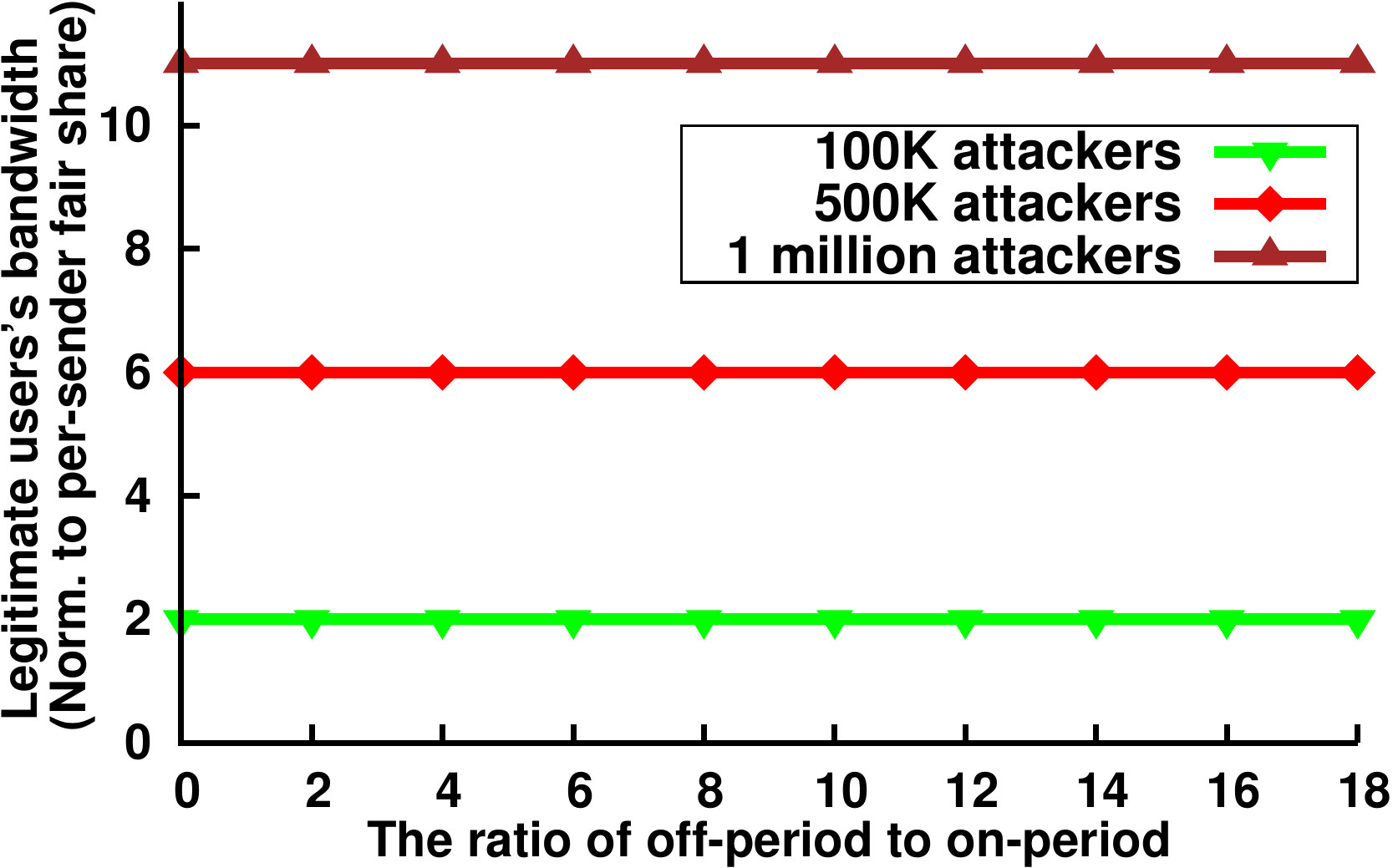}}
    \subfigure[\label{fig:simulation:b}Varying the volume of attack traffic.]{\includegraphics[scale=0.35]{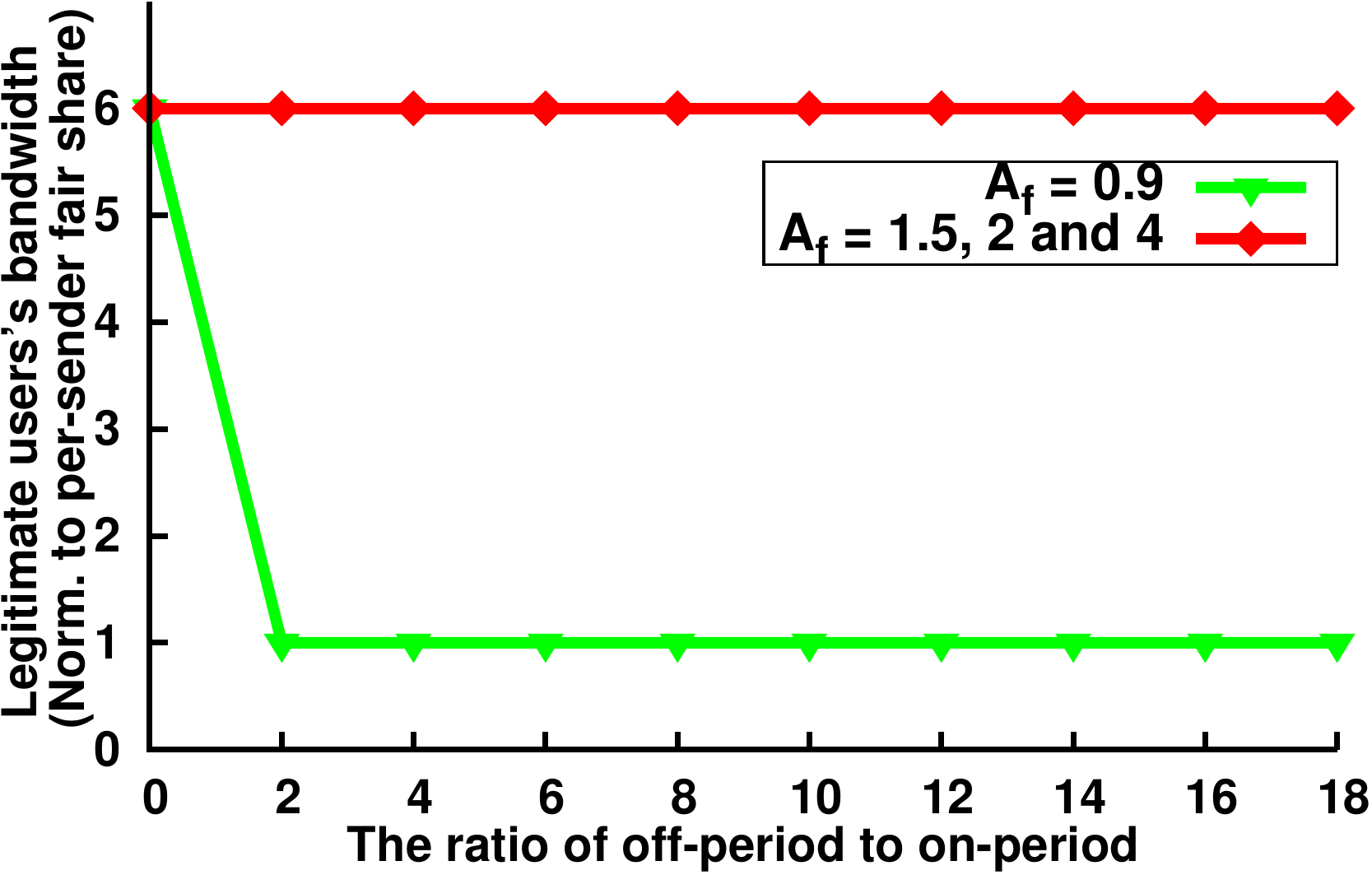}}
     \subfigure[\label{fig:simulation:c}Attackers' optimal strategy.]{\includegraphics[scale=0.35]{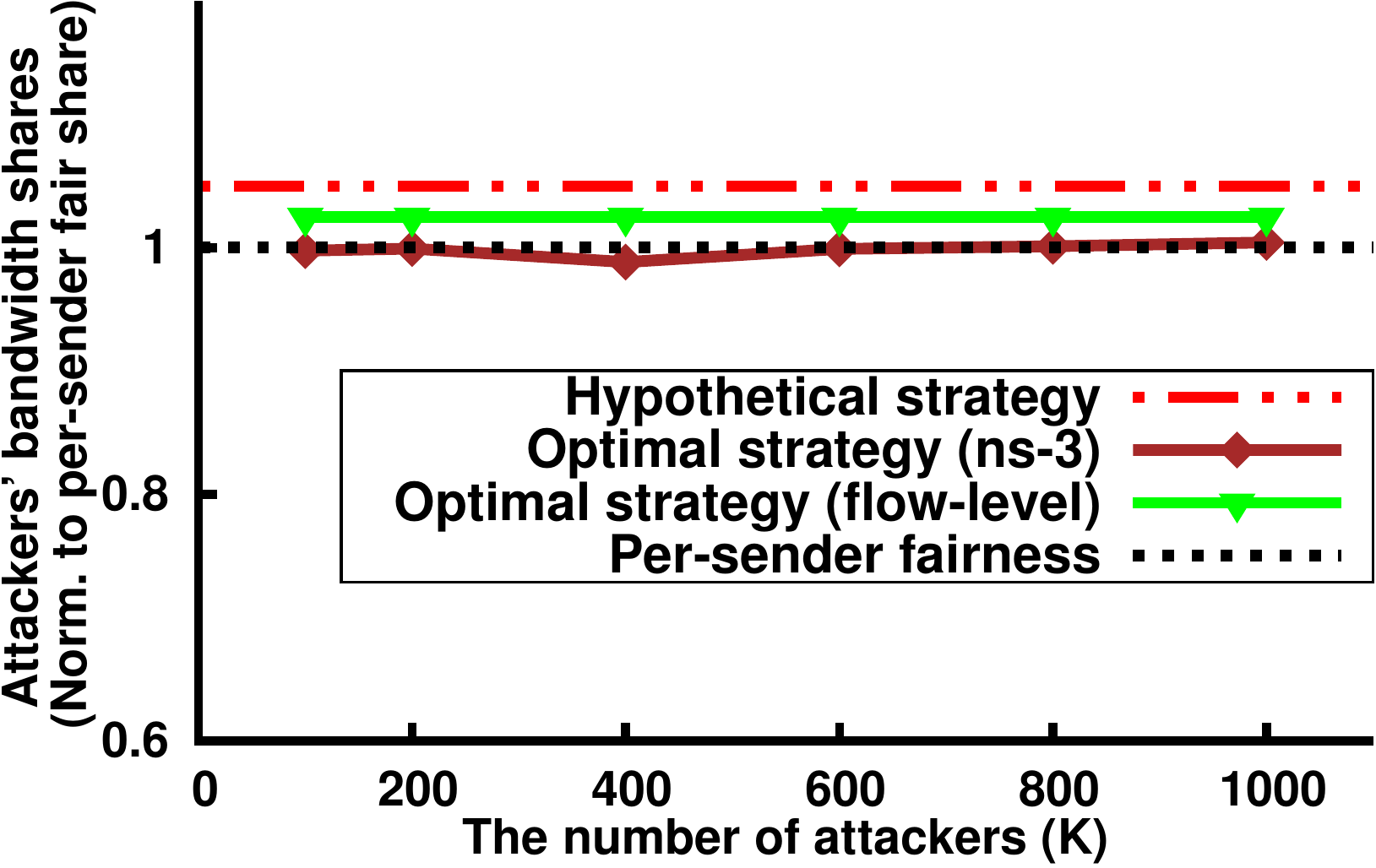}}
    }
  \caption{Mitigating large scale DDoS attacks.}
  \label{fig:simulation}
\end{figure*}

In our prototype, we implement one representative traffic policing rule for the user-specific layer:
the victim reserves bandwidth for premium clients so that they will not be affected by
DDoS attacks. Such bandwidth guarantee is achieved by the weighted fair queuing
module assigning one dedicated queue to premium clients.
We did not limit common clients' rates to ensure bandwidth shares for premium clients
because otherwise the unused bandwidth guarantee is wasted. On the contrary,
weighted fair queuing is work-conserving, allowing common clients to grab
leftover bandwidth from premium clients.
Thus the final queuing module contains three queues.

We perform three experiments on our testbed to evaluate \sys's defense, detailed as follows.

\noindent\textbf{Layer-one defense:} In this experiment, $6$ senders,
each sending $1$Gbps UDP traffic towards the victim, emulate the amplification-based
DDoS attacks in which the total volume of attack traffic is
$6\times$ the interdomain bandwidth. The $7$th sender
sends TCP traffic to represent legitimate clients. As real-life interdomain link often has
overprovisioning to absorb traffic bursts,
we set TCP flows' demand as $700$Mbps ($70\%$ of the total interdomain bandwidth).
We present our experiment results in form of sequential events, illustrated in Fig.~\ref{fig:testbed:a}.
At $t=0$s, the victim is hammered by DDoS attacks, causing complete
denial of service to legitimate clients. \sys's layer-one defense is initiated
at $t=4s$ to provide (almost) immediate DDoS prevention. Legitimate clients' bandwidth
shares grow rapidly to accommodate their traffic demand. Due to the work-conservation of
weighted fair queuing, attack traffic consumes the spare bandwidth of the interdomain link.

\noindent\textbf{Layer-two defense:} In this experiment,
although all senders are adopting TCP, $6$ of them
deviate from TCP's congestion control algorithm by continuously
injecting packets in face of congestive losses. As illustrated in
Fig.~\ref{fig:testbed:b}, malicious senders successfully exhaust the interdomain
bandwidth without the protection of \sys. At $t=4$s, \sys starts to police
traffic based on its rate limiting algorithm. As attackers fail to
comply with \sys's rate limiting, their bandwidth shares are significantly reduced,
resulting in almost zero share in the steady state. However, legitimate clients' throughput
gradually converges to their traffic demand.

\noindent\textbf{Layer-three defense:}
In this experiment, one sender is upgraded to represent premium clients with
$200$Mbps bandwidth guarantee. Another sender stands for common clients with
$500$Mbps traffic demand. The rest senders emulate attackers transmitting
malicious TCP flows. To satisfy the guarantee, the victim configures the normalized queue weight
as $0.2$, $0.7$ and $0.1$ for premium clients, common clients and UDP traffic, respectively.
Fig.~\ref{fig:testbed:c} demonstrates that premium clients' bandwidth share is
guaranteed throughout the experiment, saving them from the turbulence caused by DDoS attacks.
Further, common clients are protected after \sys enables its rate limiting, which effectively
thwarts DDoS attacks.

\subsection{Mitigating Large Scale DDoS Attacks}\label{sec:evaluation}
In this section, we evaluate \sys's defense against large scale DDoS attacks.
In the evaluation, we develop a flow-level simulator rather than completely
relying on the existing packet-level emulators or simulators (\eg Mininet~\cite{mininet}, ns-$3$~\cite{ns3})
because it takes them prohibitively long to emulate large scale DDoS attacks.
Specifically, assume that a packet-level simulator can
process one million packets per second and that one million
attack flows, each sending at $5$Mbps, attack a $10$Gbps
link. Even if we set the packet size as the maximum allowed
size $1.5$KB, it will take the simulator around $140$ hours to simulate
just one second of the attack. By concealing the detailed
per-packet processing and focusing on per-flow
behaviors, our flow-level simulator is still able to accurately
evaluate \sys, which in fact relies on flow-level states to police traffic.
However, we also perform a moderate scale simulation
on ns-$3$ to benchmark our flow-level simulator.

The network topology adopted in  simulations is similar to that of the testbed experiments
except the number of senders can be more than $1$ million and we scale up the
interdomain bandwidth to $10$Gbps.
Unless otherwise stated, the following experiments
are performed on our flow-level simulator.

We design experiments for different strategies attackers may take: \first
they launch on-off shrew attacks~\cite{low-rate} to evade detection, \second vary
the volume of attack traffic and \third dynamically adjust their rates based on
packet losses. In the on-off attack, attackers coordinate with each other to send high traffic bursts during
on periods and stay inactive during off periods. We use the ratio of the off-period's
length to the on-period's length (denoted by $\mathcal{R}^{off}_{on}$) to represent attackers' strategy in the on-off attack.
In the second strategic attack, we define the \emph{aggressiveness factor} $\mathcal{A}_f$ as the
ratio of the total volume of attack traffic to the interdomain bandwidth. Attackers may vary
the $\mathcal{A}_f$ during attacks.
In the first two experiments, attackers disrespect packet losses and
keep injecting packets in case of severe congestive losses.
The design of these two experiments is to prove that when attackers fail to adopt the optimal strategy (the third strategy discussed below), \sys accurately throttles attack flows so that legitimate senders receive more bandwidth than their guaranteed portions. In the third strategy (the optimal one), attackers rely on their transport protocols to probe packet losses so as to adjust their rates to honor the congestion. In this case, we demonstrate that \sys guarantees per-sender fairness for legitimate senders.

In the first strategic attack, we set
the length of the on-period the same as $\mathcal{D}_p$ ($5$s)
and vary the ratio $\mathcal{R}^{off}_{on}$ from $0$ to $18$.
Meanwhile we set the number of legitimate clients $\mathcal{N}_{L}{=}100K$ and
vary the number of attackers from $100K$ to $1$ million. Further, we set
$\mathcal{A}_f{=}2$ but varying the rate of each attacker based on a
Gaussian distribution.\footnote{Assume the aggregated rate of attackers is $R$,
then the Gaussian distribution's mean is $R/\mathcal{N}_A$ and the standardization is $1$.}
The experimental results, illustrated in Fig. \ref{fig:simulation:a}, show that
legitimate users can obtain more bandwidth than the per-sender
fair share regardless of $\mathcal{R}^{off}_{on}$'s value and the attack scale.
This is because \sys's rate limiting algorithm incorporates flows' previous
packet losses while making rate limiting decisions. Consequently,
even completely staying inactive during off-periods,
attackers fail to save their reputation by
the strategic on-off attack. Further, \sys can effectively
distinguish misbehaved flows from legitimate ones
no matter how many misbehaved flows are involving. Ironically,
larger attack scales result in higher benefit gains for legitimate users in the sense
that their bandwidth shares are boosted to higher levels compared with the
per-sender fair share. In all scenarios, attackers' bandwidth shares are limited to almost zero.

In the second strategic attack, we vary $\mathcal{A}_f$ from $0.9$ to $4$.
Although setting $\mathcal{A}_f{<}1$ cannot completely disconnect
the victim from its ISP, attackers can throttle legitimate users
to a tiny fraction of the total bandwidth by aggressively injecting packets
(image an analogical situation where a $900$Mbps UDP flow
competes with TCP flows on a $1$Gbps link).
Attackers experience low packet loss rates as legitimate users
cut their rates dramatically. To defend against such ``moderate"
attacks, the victim can configure \sys to start traffic policing when the
link utilization exceeds a pre-defined threshold (\eg 90\%).
We fix $\mathcal{N}_L{=}100K$ and $\mathcal{N}_A{=}500K$ in this experiment.
The results (Fig. \ref{fig:simulation:b}) show that legitimate users get
at least the per-sender fair share in all settings.
Note that in moderate attacks, attackers can prevent their bandwidth shares from being
further reduced by extending the off-period, resulting in per-sender fairness.
However, increasing $\mathcal{A}_f$ actually puts attackers in a bad situation that
their flows are blocked.

The previous two experiments prove that when attackers fail to comply with
\sys's rate limiting, their bandwidth shares are significantly reduced.
Legitimate users may therefore obtain more bandwidth than the per-sender fair share.
In the third strategy, attackers actively adjust their rates based on
packet losses so as to maintain low packet loss rates. Besides our flow-level
simulator, we also adopt the ns-$3$~\cite{ns3} in this setting.
To circumvent ns-$3$'s scalability problem, we adopt the similar approach used in NetFence~\cite{netfence}.
Specifically, we fix the number of nodes ($500$ attackers and $50$ legitimate nodes in our experiments)
and scale down the link capacity to simulate the large scale attacks. By varying the link bandwidth
from $5$Mbps to $50$Mbps, we are able to simulate the attack scenarios where
$100$K to $1$ million attackers try to flood the $10$Gbps interdomain link.
We use the ns-$3.21$ version and add \sys's rate limiting logic to the \textsl{PointToPointNetDevice} module, which
performs flow analysis upon receiving packets. To execute their optimal strategy,
attackers have to rely on TCP-like protocols to probe the network condition and determine their rates accordingly.
We test all supported TCP congestion control algorithms in ns-$3$: Tahoe, Reno, and NewReno.
Legitimate clients are adopting the NewReno.

The results (Fig. \ref{fig:simulation:c}) show that complying
with \sys's rate limiting grants attackers the per-sender fairness (results for different
TCP protocols in ns-$3$ are very close and we plot the results for the NewReno).
As stated in Lemma~\ref{lemma:attack}, the hypothetical strategy for attackers,
assuming that they are able to
know their exact allowed rates and control remote packet losses, produces
an unreachable upper bound for their bandwidth shares.
Further, our flow-level simulator and the packet-level ns-$3$ simulator
share (almost) the same results.


\section{Related Work}\label{related}
In this section, we discuss related work that has inspired the design of \sys.
Generally speaking, we categorize the previous DDoS defense approaches
into two major schools (\ie filtering-based and capability-based approaches), whereas
there are other approaches built on different defense primitives.

Filtering-based systems (\eg IP Traceback~\cite{practicalIPTrace, advancedIPTrace}, AITF~\cite{AITF},
Pushback \cite{pushback, implementPushback}, StopIt~\cite{StopIt})
stop DDoS attacks by filtering attack flows.
Thus they need to distinguish attack flows from
legitimate ones. For instance, IP Traceback uses a packet marking
algorithm to construct the path that carries attack flows so as to block
them. AITF aggregates all traffic traversing the same series of
ASs as one \emph{Flow} and blocks such flows if the victim suspects attacks.
Pushback informs upstream routers to block certain type of traffic.
StopIt assumes the victim can identify the attack flows.
However, filtering-based systems often require remote ASs to block attack traffic on the victim's behalf,
which is difficult to enforce in the Internet.
Further, these systems may falsely block legitimate flows
since the method used to distinguish attack flows could have a high false positive rate.

The capability-based systems, such as SIFF~\cite{siff} and TVA~\cite{TVA}, try to suppress attack traffic by only accepting packets carrying valid capabilities. The original design is vulnerable to the DoC attack~\cite{doc}, which can be mitigated by the Portcullis protocol~\cite{portcullis}.
NetFence \cite{netfence} is proposed to achieve network-wide per-sender fairness
based on capabilities. However, these approaches assume universal capability deployment.
CRAFT \cite{craft} and Mirage~\cite{mirage} are proposed towards real-world
deployment. CRAFT emulates TCP states for all traversing flows so that no one
can obtain a greater share than what TCP allows.
However, CRAFT requires upgrades of both the Internet core and end-hosts.
Mirage~\cite{mirage}, a puzzle-based solution, needs to be
incorporated into IPv6 deployment. The state-of-the-art in this category MiddlePolice~\cite{MiddlePolice} is readily deployable in the current Internet. However, it still relies on cloud infrastructure to police traffic, which may be privacy-invasive for some organizations.

Other DDoS defense solutions, besides the above two categories,
include SpeakUp~\cite{speakup}, Phalanx \cite{phalanx}, SOS
\cite{sos} and few future Internet architecture proposals like XIA~\cite{xia} and SCION~\cite{scion}.
SpeakUp allows legitimate senders to increase their rates to compete with attackers.
Such an approach is effective when the bottleneck happens at the application layer so that
legitimate users can get more requests processed given all their requests can be delivered.
In the case where network is the bottleneck, SpeakUp may potentially congest
the network. Phalanx and SOS propose to use large scale overlay networks to defend DDoS attacks.
XIA and SCION focus on building the clean-slate Internet architecture so as to enhance
Internet security, \eg enforcing accountability~\cite{aip}.

In contrast to these prior work, \sys is motivated to address a real-world threat and achieves two critical features (\ie deployability and privacy-preserving)  towards this end.

\section{Conclusion and Future Works}\label{conclusion}
This paper presents the design, implementation and evaluation of \sys, a new DDoS defense mechanism enabling ISPs to offer readily deployable and privacy-preserving DDoS prevention services. To provide effective DDoS prevention, \sys merely requires independent deployment at the victim's ISP and no Internet core or end-hosts upgrades, making \sys immediately deployable. Further, \sys does not require the ISP to terminate victim's application connections, allowing the ISP to operate at network layer as usual. In its design, \sys's multi-layered defense allows \sys to stop various DDoS attacks and provides both guaranteed and elastic bandwidth shares for legitimate clients. Based on  the prototype implementation, we demonstrate that \sys is scalable to deal with large scale DDoS attacks involving millions of attackers and introduces negligible packet processing overhead. Finally, our physical testbed experiments and large scale simulations prove that Umbrella is effective to mitigate various strategic DDoS attacks.

We envision two major followup directions of this work in the near future. First, the user-specific layer in \sys enables a potential DDoS victim to enforce self-desired traffic control policies during DDoS mitigation. However, one challenge is how to guide the victim to develop reasonable policies that are most suitable for its business logic. This is because proposing valid policies may require profound understanding of the victim's network traffic, which typically depends on comprehensive traffic monitoring and analysis. Unfortunately, the potential victim may lack such capability in this regard. Thus, designing and implementing various machine learning based traffic discovery tools is part of our future work. The second potential research direction is to enable smart payment between ISPs and potential victims. The high level goal is to ensure that ISPs and victims can unambiguously agree on certain filtering services so that the ISPs are paid properly on each attack packet it filters and meanwhile a potential victim can reclaim its payment back if an ISP fails to stop attacks. We propose to design a smart-contract based system in this regard, relying on the ``non-stoppable'' features of smart contracts. Our initial proposal is under review.


\balance
\small
\bibliographystyle{ieeetr}
\bibliography{paper}

\begin{IEEEbiography}[{\includegraphics[width=1in,height=1.25in,clip,keepaspectratio]{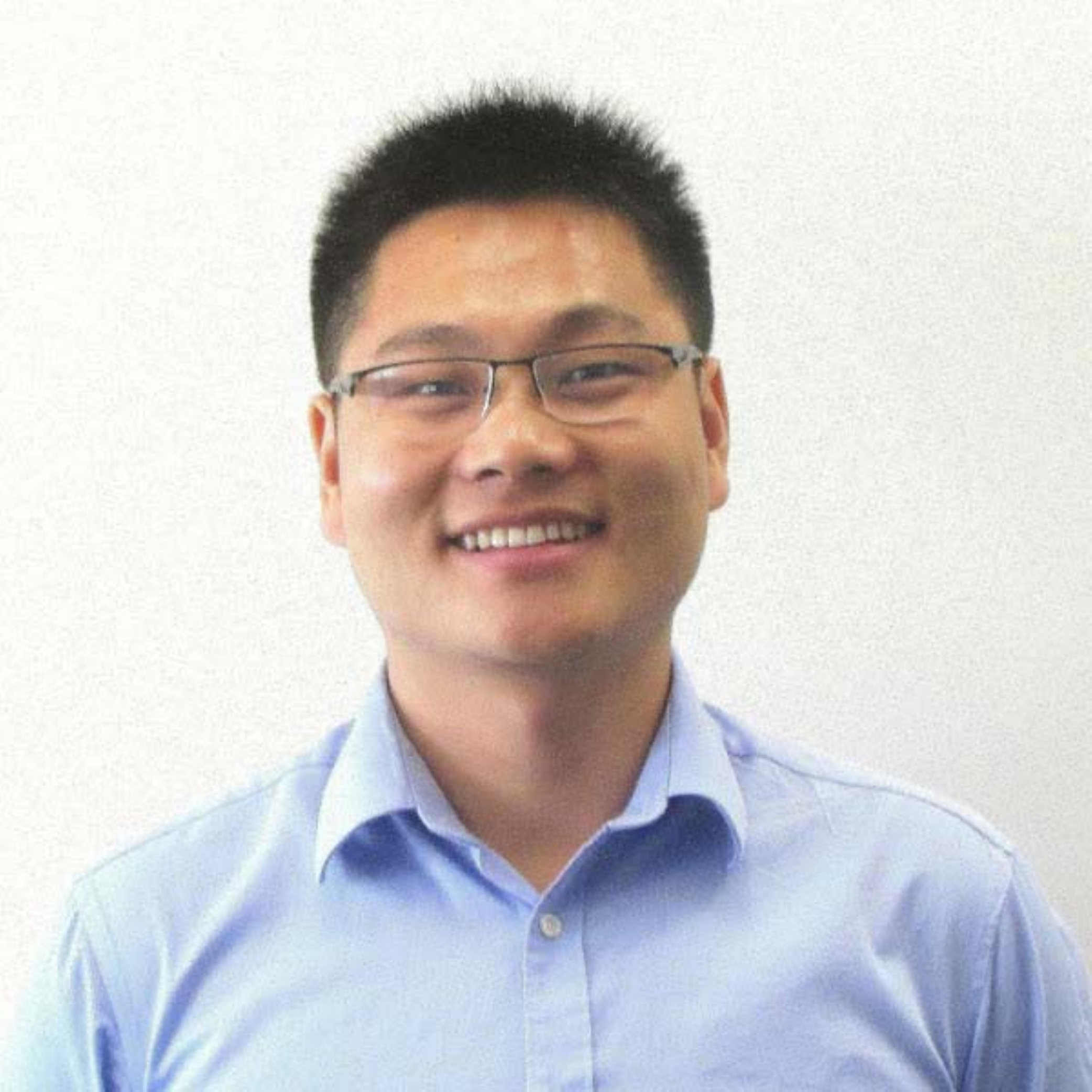}}]{Zhuotao Liu}
received his Ph.D. degree from University of Illinois at Urbana-Champaign in 2017, and B.S. degree from Shanghai Jiaotong University in 2012. Currently, he works in Network Infrastructure Team at Google, maintaining Google's global-scale private WAN. His research interests include Internet security \& privacy, data center networking and blockchain infrastructure.
\end{IEEEbiography}

\begin{IEEEbiography}[{\includegraphics[width=1in,height=1.25in,clip,keepaspectratio]{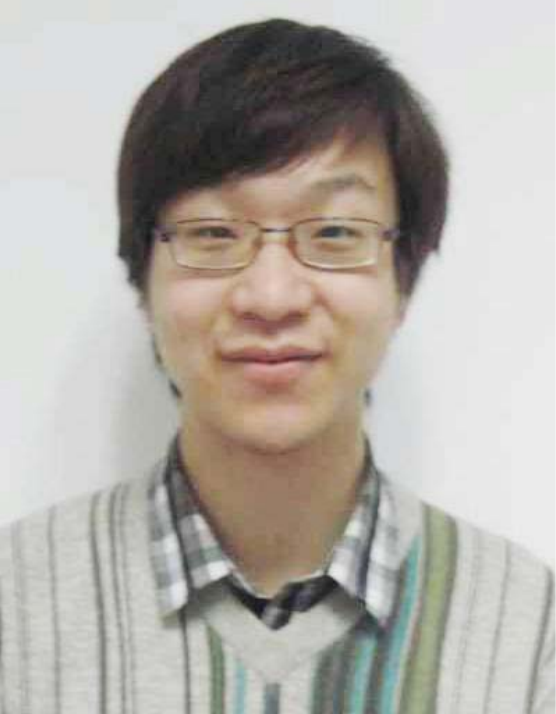}}]{Yuan Cao}(S'09-M'2014)
received his B.S. degree from Nanjing University, M.E. degree from Hong Kong University of Science and Technology and Ph.D. degree from Nanyang Technological University in 2008, 2010 and 2015, respectively. Currently he works as an assistant professor in College of Internet of Things Engineering of Hohai University. His research interests include hardware security, silicon physical unclonable function, and analog/mixed-signal VLSI circuits and systems.
\end{IEEEbiography}

\begin{IEEEbiography}[{\includegraphics[width=1in,height=1.25in,clip,keepaspectratio]{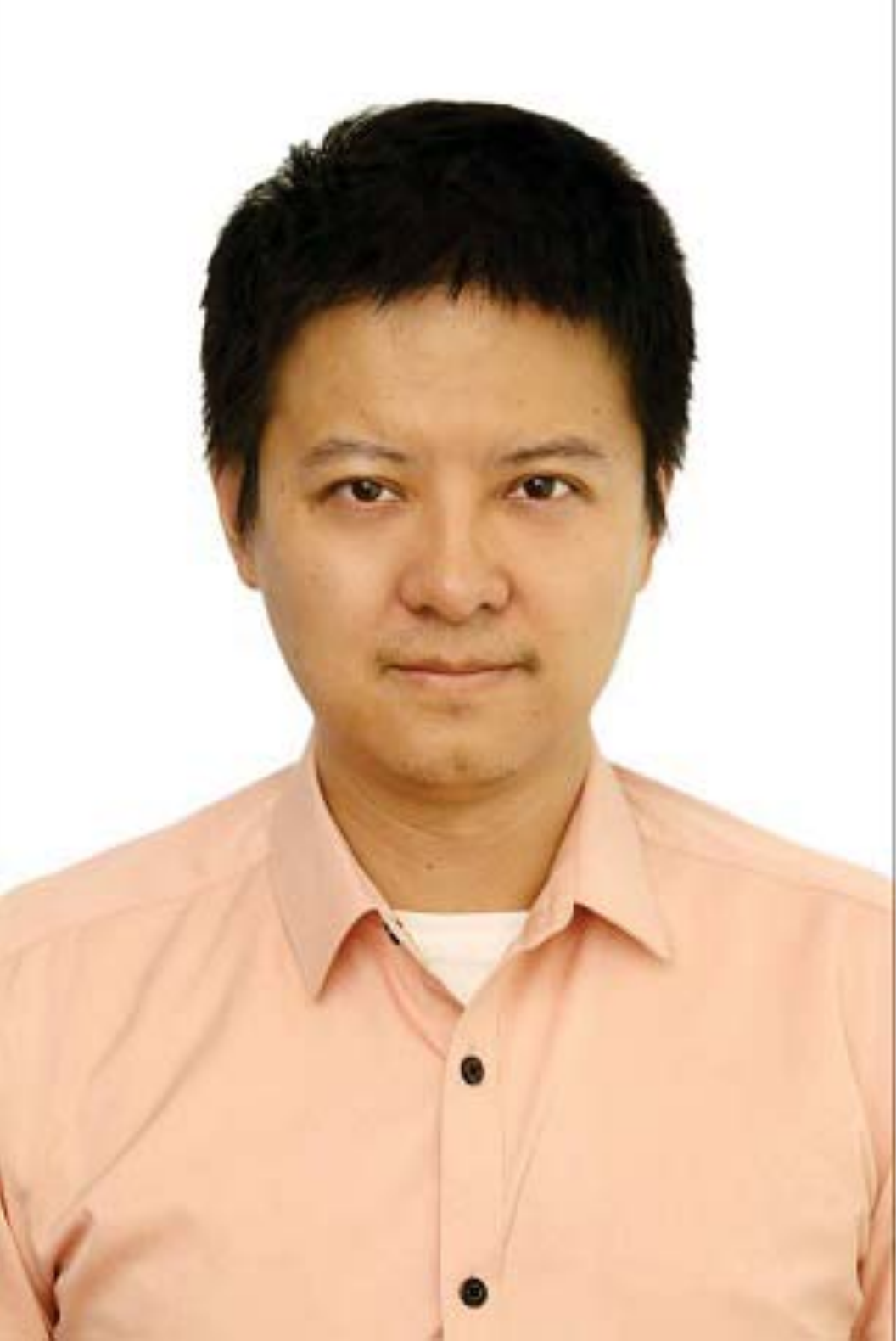}}]{Min Zhu}
received his B.S. degree and Ph.D. degree from Tsinghua University in 2006 and 2012, respectively. Currently he works as an researcher of Wuxi Research Institute of Applied Technologies Tsinghua University. His research interests include reconfigurable computing and cryptographic engineering.
\end{IEEEbiography}

\begin{IEEEbiography}[{\includegraphics[width=1in,height=1.25in,clip,keepaspectratio]{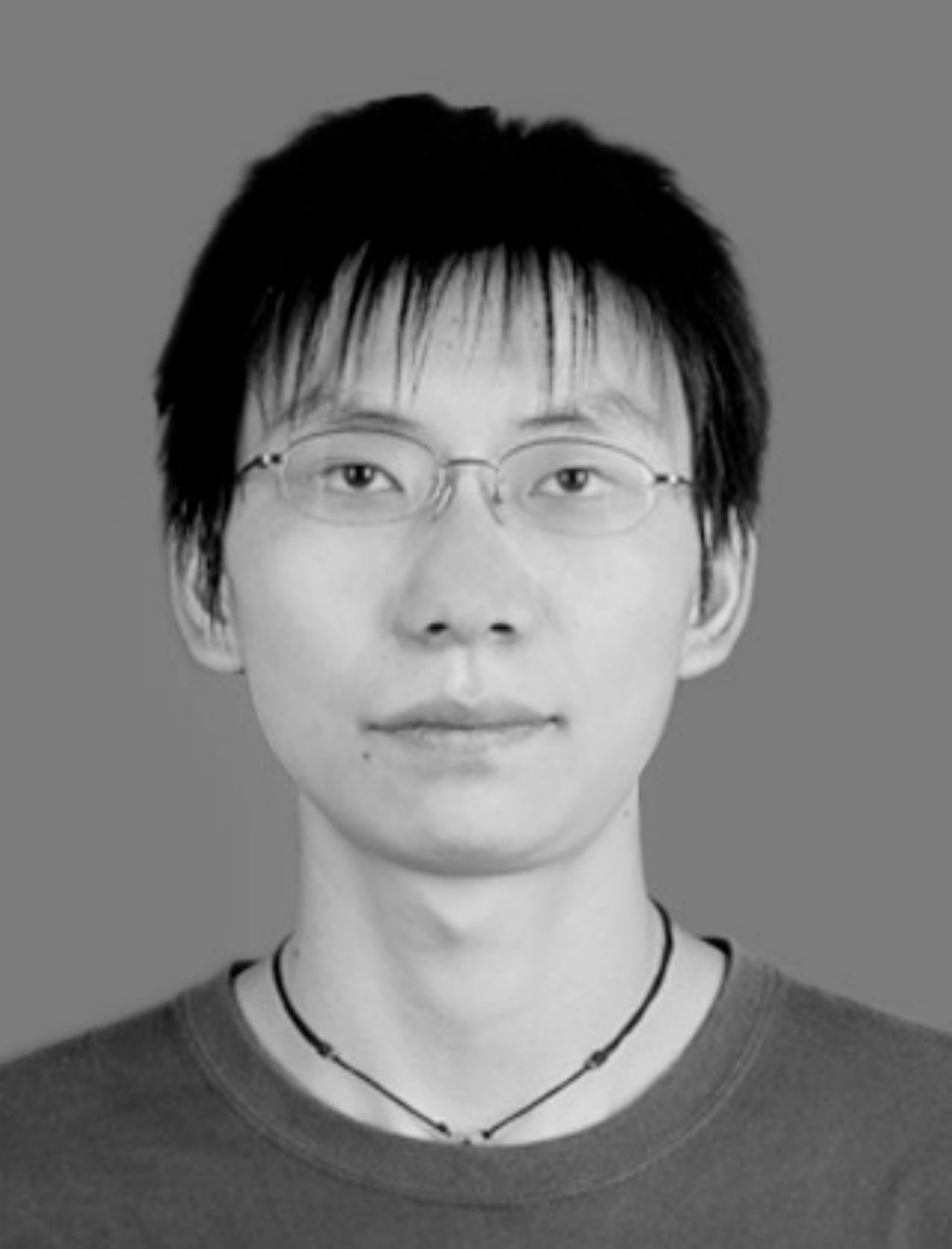}}]{Wei Ge}
received his B.S. degree and Ph.D. degree from Southeast University in 2006
and 2015, respectively. Currently he works as an assistant researcher of
Electronic Science and Engineering School of Southeast University. His research interests include reconfigurable computing, silicon physical unclonable function, and VLSI circuits and systems.
\end{IEEEbiography}

\end{document}